\documentclass[10pt,conference,compsoc]{IEEEtran}

\usepackage[utf8]{inputenc} 

\usepackage{epsfig}
\usepackage{graphicx}
\usepackage{threeparttable}
\usepackage{amsmath}
\usepackage{amsfonts}
\usepackage{bm}
\usepackage{subcaption}
\usepackage{utfsym}
\usepackage{amsthm}
\usepackage{xcolor}
\definecolor{lightgray}{rgb}{0.83, 0.83, 0.83}
\definecolor{mygray}{rgb}{0.93, 0.93, 0.93}
\usepackage{lipsum}
\usepackage{mathtools}
\usepackage[makeroom]{cancel}
\usepackage{mathrsfs}
\usepackage{enumitem}
\usepackage{multirow}
\usepackage{tabularx}
\usepackage{tikz}
\newcommand*\emptycirc[1][1ex]{\tikz\draw (0,0) circle (#1);} 

\newcommand*\fullcirc[1][1ex]{\tikz\fill (0,0) circle (#1);}

\usepackage[T1]{fontenc}    
\usepackage{booktabs}      
\usepackage{nicefrac}       
\usepackage{microtype}      
\usepackage[most]{tcolorbox}
\usepackage{caption} 
\usepackage{setspace}
\usepackage[noend]{algorithmic}
\usepackage{algorithm}

\usepackage{newfloat}
\usepackage{listings}
\DeclareCaptionStyle{ruled}{labelfont=normalfont,labelsep=colon,strut=off} 
\floatstyle{ruled}
\newfloat{listing}{tb}{lst}{}
\floatname{listing}{Listing}
\usepackage{courier}  

\def\eqref#1{equation~\ref{#1}}

\def\1{\bm{1}}

\def\vm{{\bm{m}}}

\def\vs{{\bm{s}}}

\def\vx{{\textbf{x}}}

\def\vz{{\bm{z}}}

\DeclareMathAlphabet{\mathsfit}{\encodingdefault}{\sfdefault}{m}{sl}
\SetMathAlphabet{\mathsfit}{bold}{\encodingdefault}{\sfdefault}{bx}{n}

\makeatletter
\newcommand{\distas}[1]{\mathbin{\overset{#1}{\kern\z@\sim}}}
\makeatother

\theoremstyle{definition}
\newtheoremstyle{bfnote}%
  {}{}
  {\itshape}{}
  {\bfseries}{.}
  { }{\thmname{#1}\thmnumber{ #2}\thmnote{ (#3)}}
\theoremstyle{bfnote}
\newtheorem{definition}{Definition}

\newtheorem{theorem}{Theorem}
\newtheorem*{theorem*}{Theorem}

\newtcolorbox{colorquote}[1][]{
    boxrule=0.5pt,
    left=1pt,
    right=1pt,
    top=1pt,
    bottom=1pt,
    colback=black!5,
    colframe=black!55,
    notitle,
    enhanced,
    breakable,
}
\definecolor{darkgreen}{rgb}{0.0, 0.5, 0.0}
\lstdefinestyle{mystyle}{
language=C, 
basicstyle=\ttfamily\footnotesize,
  keywordstyle=\color{blue},
  commentstyle=\color{darkgreen},
  emph={return},
  emphstyle=\color{magenta},
  numbers=none,
  frame = none,
  breaklines=true,
  breakatwhitespace=true,
  tabsize=2,
  columns=fullflexible,
  captionpos=b,
  escapechar=@
}

\newcommand{\ours}{\textsc{EnBecome}}

\usepackage[normalem]{ulem}

\newcommand{\todoc}[2]{{\textcolor{#1}{\textbf{#2}}}}
\newcommand{\todored}[1]{{\todoc{red}{\textbf{[#1]}}}}

\newcommand{\todoblue}[1]{\todoc{blue}{\textbf{[#1]}}}

\newcommand{\txhan}[1]{\todoblue{txhan: #1}}

\newcommand{\ws}[1]{\todored{WS: #1}}

\newcommand{\circled}[1]{\textcircled{\raisebox{-.9pt}{#1}}}

\usepackage{xcolor}

\usepackage{cite}
\usepackage{hyperref}
\hypersetup{
    colorlinks=true,
    linkcolor=blue,
    citecolor=blue,
}

\usepackage{url}
\usepackage{orcidlink}

\begin{document}

\title{Enhancing and Reporting Robustness Boundary of Neural Code Models\\ for Intelligent Code Understanding}

\author{
Tingxu Han$^{\orcidlink{0000-0003-1821-611X}}$,  
Wei Song$^{\orcidlink{0000-0002-1240-8412}}$,
Weisong Sun$^{\orcidlink{0000-0001-9236-8264}}$, 
Hao Wu$^{\orcidlink{0009-0001-5945-7510}}$,
Chunrong Fang$^{\orcidlink{0000-0002-9930-7111}}$, 
 \\
Yuan Xiao$^{\orcidlink{0009-0009-3166-8007}}$,
Xiaofang Zhang$^{\orcidlink{0000-0002-8667-0456}}$, Zhenyu Chen$^{\orcidlink{0000-0002-9592-7022}}$,
and 
Yang Liu $^{\orcidlink{0000-0001-7300-9215}}$

\IEEEcompsocitemizethanks{
\IEEEcompsocthanksitem Tingxu Han, Chunrong Fang, Yuan Xiao, and Zhenyu Chen are with the State Key Laboratory for Novel Software Technology, Nanjing University, Nanjing 210093, China (e-mail: \{txhan, 
dz21320004\}@smail.nju.edu.cn, \{fangchunrong, zychen\}@nju.edu.cn).
Zhenyu Chen is also with Shenzhen Research Institute of Nanjing University, China.
\IEEEcompsocthanksitem Wei Song is with the School of Computer Science and Engineering, University of New South Wales, New South Wales 2052, Australia (e-mail: wei.song1@unsw.edu.au). 
\IEEEcompsocthanksitem Weisong Sun and Yang Liu are with the School of Computer Science and Engineering, Nanyang Technological University, Nanyang 639798, Singapore (e-mail: \{weisong.sun, yangliu\}@ntu.edu.sg).
\IEEEcompsocthanksitem Hao Wu and Xiaofang Zhang are with the School of Computer Science and Technology, Soochow University, Suzhou 215006, China (e-mail: 20224207007@stu.suda.edu.cn, xfzhang@suda.edu.cn).
\IEEEcompsocthanksitem Weisong Sun is the corresponding author.
Tingxu Han and Wei Song have an equal contribution.
}
}

\maketitle
\begin{abstract}
With the development of deep learning, Neural Code Models (NCMs) such as CodeBERT and CodeLlama are widely used for code understanding tasks, including defect detection and code classification. However, recent studies have revealed that NCMs are vulnerable to adversarial examples—inputs with subtle perturbations that induce incorrect predictions while remaining difficult to detect. Existing defenses attempt to address this issue through data augmentation to empirically improve robustness, but they are costly, provide no theoretical robustness guarantees, and typically require white-box access to model internals, such as gradients.

To address the above challenges, we propose \ours{}, a novel black-box training-free and lightweight adversarial defense.
\ours{} is designed to both enhance empirical robustness and report certified robustness boundaries for NCMs. 
\ours{} operates solely during inference, introducing random, semantics-preserving perturbations to input code snippets to smooth the NCM’s decision boundaries. 
This smoothing enables \ours{} to formally certify a robustness radius within which adversarial examples can never induce misclassification, denoted certified robustness.
We conduct comprehensive experiments across multiple NCM architectures and tasks. 
Results show that \ours{} significantly reduces attack success rates while maintaining high accuracy. For example, in defect detection, it reduces the average ASR from 42.43\% to 9.74\% with only a 0.29\% drop in accuracy. Furthermore, \ours{} achieves an average certified robustness radius of 1.63, meaning that adversarial modifications to no more than 1.63 identifiers are provably ineffective. 
\end{abstract}
\section{Introduction}
\label{sec:intro}


Recently, with the rapid development of deep learning (DL), it has become more and more popular to deploy it into intelligent code understanding tasks in software engineering~\cite{2022-Use-of-DL-in-SE-Research, 2023-LLMs-for-SE-A-Literature-Review, 2023-Survey-on-LLMs-for-SE}. 
Currently, DL models are pre-trained on a large-scale code dataset and then fine-tuned or prompted with a few task-specific data. 
These large pre-trained models or fine-tuned models (including large language models, LLMs) for code understanding can be collectively referred to as Neural Code Models (NCMs). 
Existing software systems are built on millions of lines of code~\cite{shin2010evaluating}, providing sufficient training data to pre-train NCMs.
Leveraging this extensive code corpus, NCMs like CodeBERT~\cite{feng2020codebert} and CodeLlama~\cite{roziere2023code} showcase exceptional performance on code understanding tasks, such as defect detection~\cite{wu2023large, zhou2019devign}.
Developers deploy these models by fine-tuning or prompting them with task-specific datasets.

\begin{figure}[t]
\centering
    \subfloat[The original code snippet.]
    {
        \includegraphics[width=0.836\linewidth]{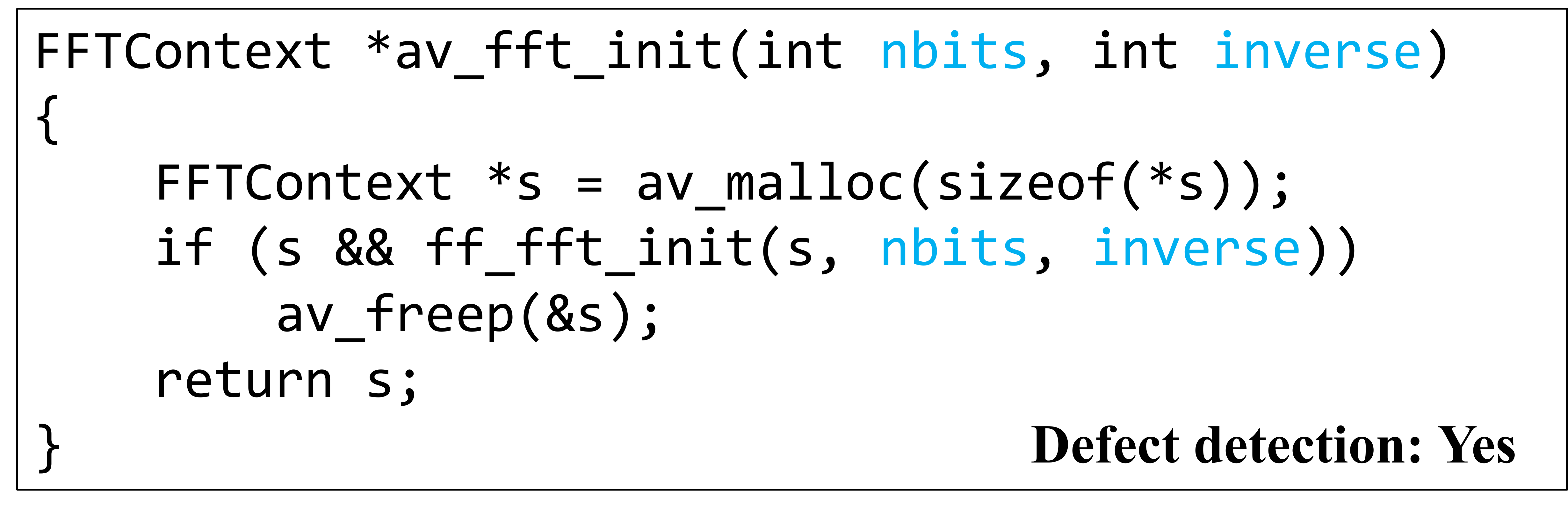}
        \label{fig:ori_sample}
    }
    \vfill
    \subfloat[The adversarial code snippet.]
    {
        \includegraphics[width=0.836\linewidth]{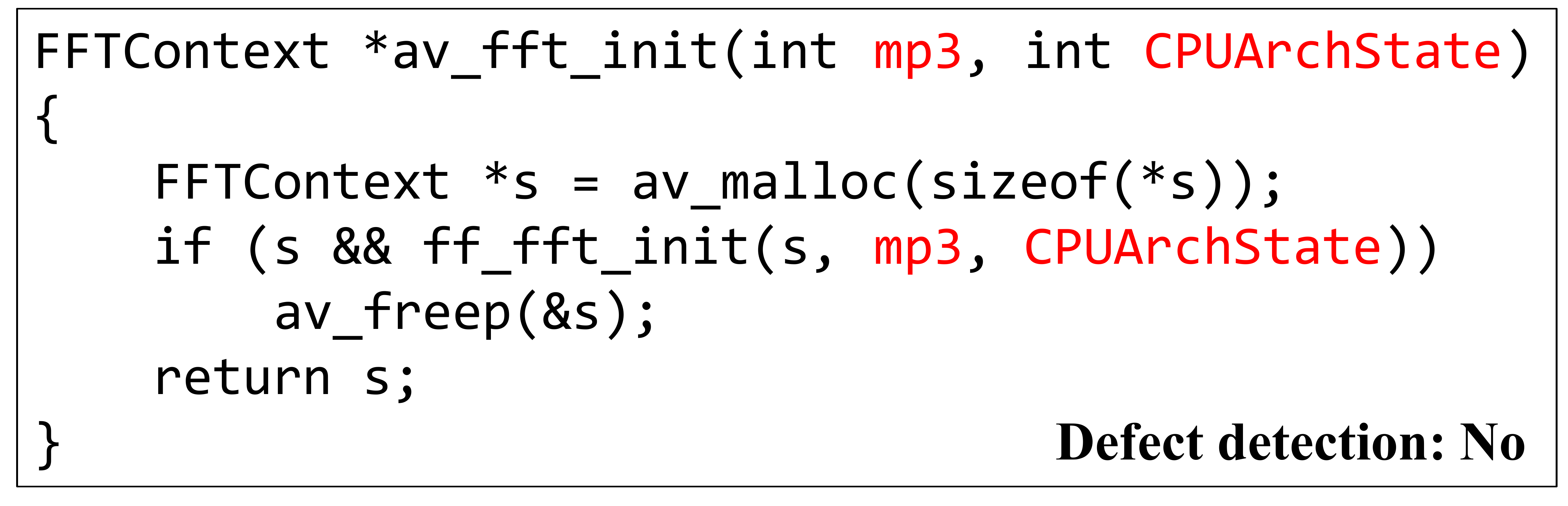}
        \label{fig:adv_sample}
    }
    \vspace{-5pt}
    \caption{An adversarial example on CodeBERT in defect detection. The adversarial code snippet (ACS) is crafted by only modifying the identifiers, but it leads to a misclassification from ``\texttt{Yes}'' to ``\texttt{No}''. The adversarial code snippet still contains the same defect as the original one.
    }
    \label{fig:case_ori_adv_example}
    \vspace{-8pt}
\end{figure}

\begin{figure*}[t]
    \centering
    \includegraphics[width=0.8\linewidth]{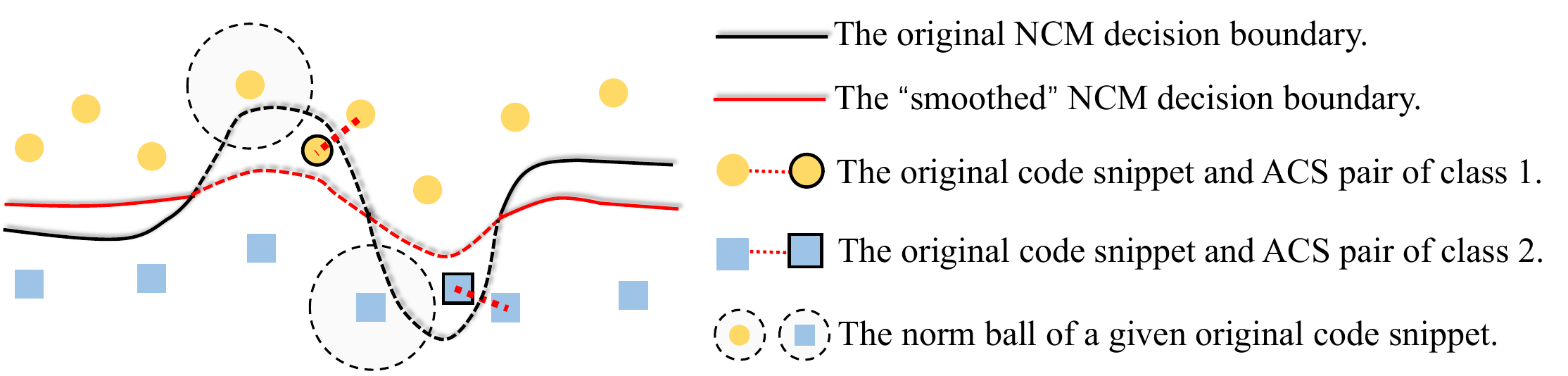}
    \caption{An intuitive illustration of \ours{}'s scenario and objective in defect detection. 
    The adversary crafts ACSs based on original code snippets to cross the original decision boundary, causing misclassification. The norm measures the difference between the original code snippets and their corresponding ACSs. \ours{} introduces controlled randomness into the input space during inference, averaging predictions over perturbed inputs within a norm ball to achieve a ``smoothed'' decision boundary (from the black line to the red line).}
    \label{fig:intuitive_overview}
    \vspace{-7pt}
\end{figure*}

However, recent studies~\cite{2024-CodeLM-Security, bielik2020adversarial, chen2017adversarial, jha2023codeattack, liu2021practical, quiring2019misleading, severi2021explanation} have demonstrated that NCMs are vulnerable to adversarial examples (AEs).
An adversary crafts AEs by making subtle perturbations on the identifiers of code snippets. 
The perturbations are always imperceptible for humans but significant for NCMs, inducing misclassification~\cite{2022-ALERT, 2020-MHM,2023-CODA}.
Figure~\ref{fig:case_ori_adv_example} showcases an intuitive example.  
The original code snippet in Figure~\ref{fig:ori_sample} contains a defect where the function may return a dangling pointer. NCM CodeBERT correctly identifies this and predicts ``\texttt{Yes}''. In contrast, the adversarial code snippet in Figure~\ref{fig:adv_sample}, which modifies only the identifiers, causes the model to misclassify it as ``\texttt{No}'', even though the defect is still present.

Various techniques have been developed to enhance the adversarial robustness of NCMs by leveraging data augmentation and post-training. These methods typically begin by augmenting the training dataset and then use the augmented data in a post-training phase to improve robustness.
For instance, RoPGen~\cite{2022-ropgen} implements data augmentation by utilizing automatic coding style transformations to create robust coding style patterns that are difficult for attackers to manipulate. 
SPACE, on the other hand, generates worst-case, semantic-preserving examples in the continuous embedding space rather than the discrete token space.
Both approaches incorporate the augmented samples into the training pipeline, forming the basis of adversarial training. 
However, adversarial training is computationally expensive and requires full access to the model's parameters, making it impractical for large-scale NCMs, such as StarCoder~\cite{2023-StarCoder} and CodeLlama~\cite{roziere2023code}.

To handle the above weaknesses, we propose \ours{} (\underline{En}hancing and reporting robustness \underline{B}oundary of n\underline{E}ural \underline{CO}de \underline{M}odels for  Intelligent Code Und\underline{E}rstanding), a novel black-box and training-free method that enhances the robustness of NCM for intelligent code understanding and reports certified robustness boundaries.
Unlike prior approaches, \ours{} operates entirely at inference time and does not require access to model parameters or additional training.

Intuitively, \ours{} hardens robustness by ``smoothing'' the decision boundaries of a given NCM.
Specifically, \ours{} introduces random, semantics-preserving perturbations to generate multiple variants of the input and aggregates their predictions through majority voting. 
This process effectively causes the model to perceive each input as a distribution over a local neighborhood, rather than as a single point.
This makes the model less sensitive to small adversarial changes and more resistant to attacks.
Figure~\ref{fig:intuitive_overview} illustrates the intuition.
The black line represents the original decision boundary, which adversarial examples can easily cross by making small changes to the input. By contrast, the red line represents the ``smoothed'' boundary achieved by \ours{}, which generalizes better and requires an adversary to introduce significantly larger perturbations to succeed.
Although many random smoothing techniques have been proposed in vision~\cite{li2023sok,cohen2019certified,fang2025multi} and NLP~\cite{ye2020safer,zhang2024text,cevher2025certified}, they operate in continuous or token-based spaces.
Previous randomized smoothing techniques in vision~\cite{li2023sok, cohen2019certified, fang2025multi} and NLP~\cite{ye2020safer, zhang2024text, cevher2025certified} are designed for continuous or token-based input spaces. Directly applying these methods to code area often breaks syntax and disrupts semantic structure. 
In contrast, \ours{} introduce semantics-preserving smoothing for NCMs in the discrete and syntax-sensitive $L_0$ space, achieving certified robustness while preserving code correctness.
The contribution of this paper can be summarized as follows:
\vspace{-2pt}
\begin{itemize}[leftmargin=8pt]
    \item We propose a black-box, training-free approach \ours{}, smoothing the decision boundary and improving the robustness of NCMs for intelligent code understanding.
    \item To the best of our knowledge, this is the first work to establish certified robustness for neural code models. We introduce a semantics-preserving smoothing technique and derive a theoretical robustness bound, guaranteeing that no adversarial example can succeed within this bound.
    \item We conduct comprehensive evaluations of \ours{} regarding effectiveness, time cost, and generalization. 
    \item We make all the implementation code of \ours{} and datasets used in our paper publicly available~\cite{2024-Our-Artifacts}.
\end{itemize}

\section{Background and preliminary}
\label{sec:background_preliminary}
NCMs interpret code using neural language models~\cite{palacio2024toward}. For instance, neural language models such as BERT~\cite{kenton2019bert} and T5~\cite{2020-T5} are fine-tuned and adapted specifically for intelligent code understanding tasks, resulting in NCMs like CodeBERT~\cite{feng2020codebert} and CodeT5~\cite{2021-CodeT5}.
In this section, we formalize the background of NCMs and intelligent code understanding tasks at first.
Subsequently, we introduce the definition of adversarial examples (AEs) that \ours{} aims to address and formalize \ours{}'s objectives.


\subsection{Code Understanding and NCMs}
\label{subsec:ncms_code_understanding}
\begin{figure*}[t]
    \centering
    \includegraphics[width=0.75\linewidth]{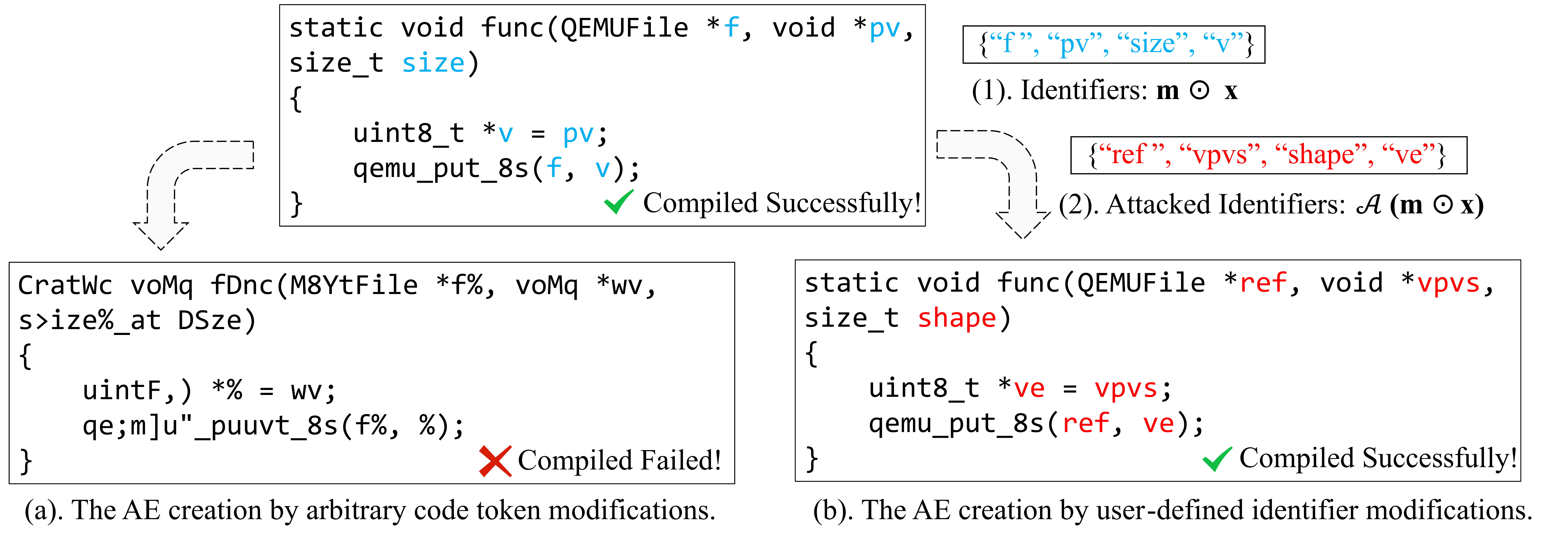}
    \caption{An intuitive illustration of an AE creation. Arbitrary code token modifications to create an AE cause syntax errors, making it uncompilable (a). User-defined identifier modifications preserve code syntax correctness and keep compilable (b).}
    \label{fig:AE_creation}
    \vspace{-5pt}
\end{figure*}
Code understanding is a challenging task that requires developers to assimilate extensive information related to code semantics and domain-specific concepts.
With the advancement of deep learning techniques, NCMs have been extensively utilized to effectively address various code understanding tasks, such as defect detection~\cite{zhou2019devign, wang2016automatically}, clone detection~\cite{2020-Functional-Code-Clone-Detection, wei2017supervised}.
Formally, given a dataset $\{\mathcal{X}, \mathcal{Y}\}$, an NCM $f(x)$ maps an input code snippet $\vx:=\{x_1, x_2, \dots, x_N\} \in \mathcal{X}$ to the ground truth label $c \in \mathcal{Y}$:
\begin{equation}
    f(\vx) = c, \vx \in \mathcal{X}, c \in \mathcal{Y}
    \label{eq:code_understanding}
\end{equation}

\subsection{Adversarial Examples (AEs)}
\label{subsec:adversarial_examples}
Adversarial examples (AEs) are inputs intentionally modified to cause the model to make incorrect predictions. In natural language tasks, adversaries can create AEs by changing arbitrary tokens~\cite{gao2018black, 2020-bertattack, 2020-TextFooler}. However, this approach is infeasible in code understanding due to strict syntax constraints: arbitrary token changes can easily lead to syntax errors, as shown in Figure~\ref{fig:AE_creation}a.
To ensure syntactic validity, adversarial attacks on code typically modify only user-defined identifiers~\cite{2020-MHM, 2022-ALERT, 2023-CODA, jha2023codeattack}, which preserves the structure and compilability of the program, as illustrated in Figure~\ref{fig:AE_creation}b.

Formally, we define a binary mask vector $\vm$ over a code snippet $\vx = \{x_1, x_2, \dots, x_N\}$ to indicate the positions of identifiers: $m_i = 1$ if $x_i$ is an identifier, and $m_i = 0$ otherwise.
Adversaries perturb $k$ identifiers based on $\vm$ to generate an AE $\vx'$, inducing $f(\vx') \ne c$.
The AE set $\mathcal{X}^*$, which means all possible AEs crafted by $\vx$, is formalized as:
\begin{equation}
\begin{aligned}
    & f(\vx') \ne c,\quad \vx' \in \mathcal{X}^*, \\
    & \mathcal{X}^* := \left\{ \vx' \mid \vx' = (1-\vm) \odot \vx + \mathcal{A}(\vm \odot \vx)     \right\}
\end{aligned}
\label{eq:formal_AE}
\end{equation}
Here, $\odot$ denotes element-wise multiplication, and $\mathcal{A}(\cdot)$ is an adversarial transformation applied to the extracted identifiers. 
Note that this transformation is consistent, meaning that each identifier is globally modified in a uniform and syntax-preserving manner across its occurrences in the code.
The final AE $\vx'$ combines the unmodified tokens $(1 - \vm) \odot \vx$ with the perturbed identifiers $\mathcal{A}(\vm \odot \vx)$.

\subsection{Objective of \ours{}}
\label{subsec:objective}
Existing defenses such as SPACE~\cite{2022-SPACE} and RoPGen~\cite{2022-ropgen} rely on costly post-training and lack theoretical guarantees, making them impractical and unreliable as NCMs scale. To overcome these limitations, we propose \ours{}, a training-free, black-box method that operates entirely at inference time. \ours{} improves empirical robustness while certifying a provable robustness boundary within which no successful adversarial example can exist.
In this section, we formalize the objectives of empirical and certified robustness, providing both theoretical definitions and intuitive insights.

\begin{definition}[Empirical Robustness]
    Empirical robustness means that NCMs should maintain consistent predictions for original samples and the corresponding AEs, which is formalized as follows:
    \begin{equation*}
        \centering
        \forall \vx \in \mathcal{X}, \vx' \in \mathcal{X}^*, f(\vx)=f(\vx'),
    \end{equation*}
    where ($\vx$, $\vx'$) means the pair of the original sample and its corresponding AE, $\mathcal{X}$ the original code snippet set and $\mathcal{X}^*$ the adversarial example set (defined by Eq~\ref{eq:formal_AE}). 
    \label{def:empirical_robustness}
\end{definition}
Most of the existing adversarial defense techniques focus on enhancing the empirical robustness through data augmentation~\cite{2022-SPACE} or disturbing embedding space~\cite{2022-ropgen}, claiming improved robustness against AEs.
However, this robustness is based on empirical results and does not provide a theoretical bound for future attacks.
In that case, we further introduce the definition of \textit{certified robustness}.


\begin{definition}[Certified Robustness]
    Certified robustness of an NCM around a given input $\vx$ means that, within a calculated theoretic bound, the NCM will always maintain consistent predictions for $\vx$ and the corresponding AE $\vx'$. This is formalized as the following:
    \begin{equation}
       \forall \vx' \in \mathcal{X}^*, f(\vx)=f(\vx')=c, c \in \mathcal{Y}, s.t. |\vx \oslash \vx'| \leq r.
    \end{equation}
    where $c$ denotes the ground truth label and $\mathcal{X}^*$ means all possible AEs crafted by the given original sample $\vx$.
    $\vx \oslash \vx'$ indicates the different identifier set between $\vx$ and $\vx'$.
    We define the theoretical bound $r$ as the certified robustness radius.
    \label{def:certified_robustness}
\end{definition}

Intuitively, certified robustness means that the model’s prediction remains unchanged when an adversary modifies up to $r$ identifiers in the input. In this case, no adversarial example within this bound can successfully fool the model.
Based on this definition, the objective of \ours{} is to construct an NCM $g(\cdot)$ that achieves both strong empirical robustness and certified robustness guarantees.


\section{Methodology}
\label{sec:meth}


\begin{figure*}[t]
    \centering
    \includegraphics[width=\linewidth]{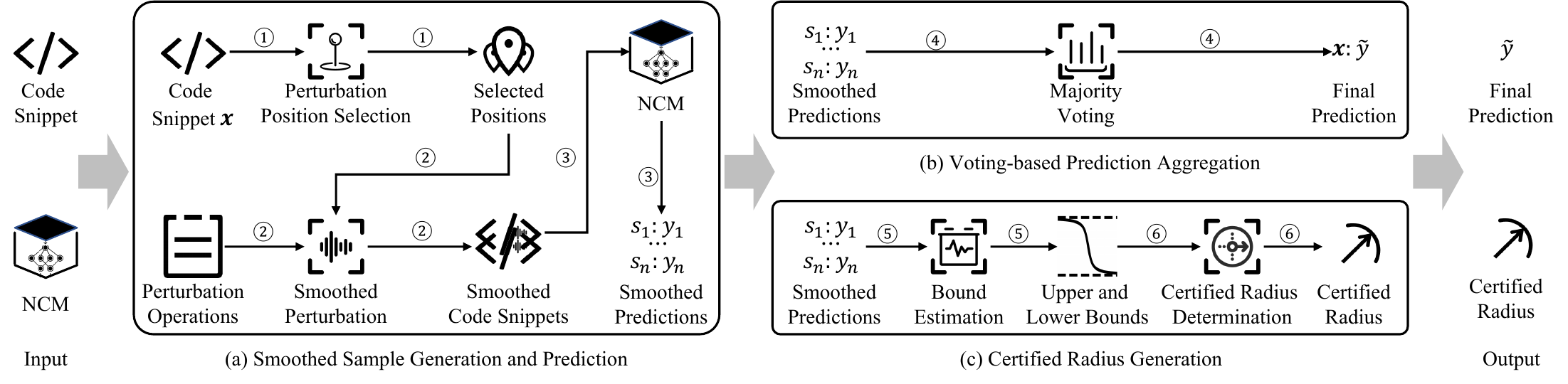}
    \caption{The workflow of \ours{}. \ours{} first generates smoothed samples and conducts predictions of them in phase (a).
    In phase (b), \ours{} aggregates the predictions and outputs the final prediction $\tilde{y}$.
    In phase (c), \ours{} generates the certified radius and reports the certified robustness.
    }
    \label{fig:workflow}
    \vspace{-18pt}
\end{figure*}

Figure~\ref{fig:workflow} illustrates the overall workflow of \ours{}. 
\ours{} is a black-box requiring no alterations to NCMs.
\ours{} reduces the success rate of adversarial attacks by smoothing NCM predictions. 
In addition, \ours{} can report the certified robustness boundary of NCMs. 
Specifically, given an input code snippet $\vx_0$ (potentially an adversarial example) and the target NCM for defense, \ours{} achieves robust prediction and reports certified robustness through three stages: (a) smoothed sample generation and prediction, (b) voting-based prediction selection, and (c) certified radius generation. 
In phase (a), \ours{} generates a set of smoothed code snippets and then collects the prediction results of the NCM for these code snippets. 
In phase (b), \ours{} aggregates the above predictions by a voting-based mechanism and outputs the final robust prediction.
Subsequently, \ours{} reports certified robustness by generating the certified radius.
\begin{table}[t]
    \centering
    \footnotesize
    \caption{The semantic-preserving operation used in our paper. Three different operations are considered.}
    \label{tab:op_list}
    \begin{tabular}{lp{3.8cm}l}
        \toprule
        Name & Description & Example \\
        
        \midrule
        
        
        --- & The original sample. & 
        \begin{lstlisting}[style=mystyle]
int f(void *env) 
        \end{lstlisting}\\
        
        \textbf{Insert} & Randomly insert a new character into the identifier. & {\raggedright\begin{lstlisting}[style=mystyle]
int f(void *enQv)
        \end{lstlisting}} \\
        
        \textbf{Replace} & Randomly replace a character in the identifier. & 
        \begin{lstlisting}[style=mystyle]
int f(void *enQ)
        \end{lstlisting} \\
        
        \textbf{Delete} & Randomly delete a character from the identifier. & 
        \begin{lstlisting}[style=mystyle]
int f(void *nv)
        \end{lstlisting} \\
        
        \bottomrule
    \end{tabular}
    \vspace{-8pt}
\end{table}

\begin{figure}
    \centering
    \includegraphics[width=\linewidth]{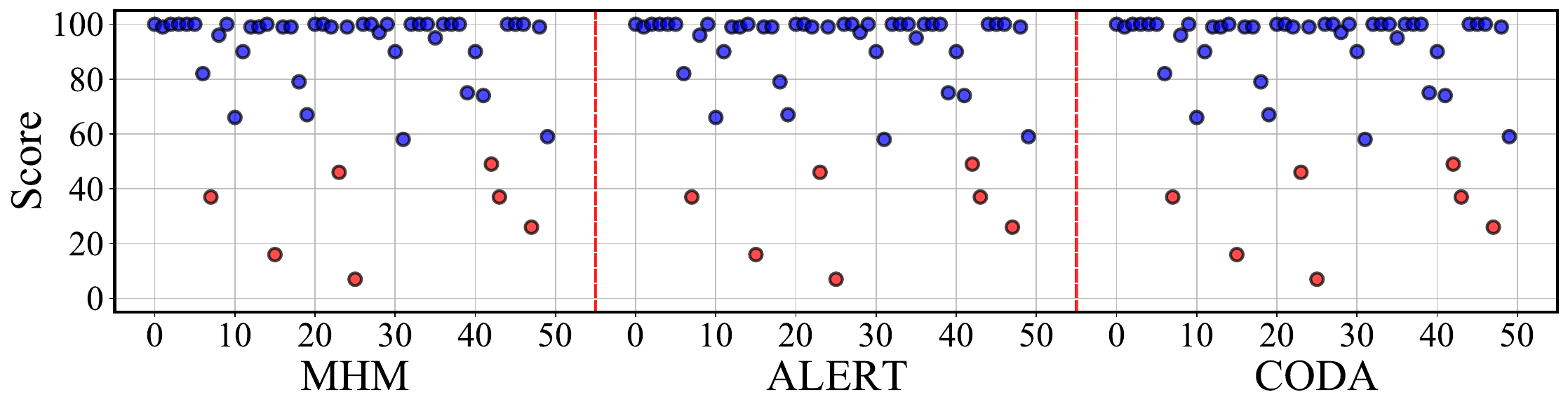}
    \vspace{-6pt}
    \setlength{\abovecaptionskip}{0pt}  
    \setlength{\belowcaptionskip}{0pt}  

    \caption{Evidence for our insight. 
   \ours{} generates $N$($N$=100) smoothed samples per code snippet, with the score indicating how many are predicted as the ground-truth label. 
   The color represents the final prediction $\tilde{y}$, where blue indicates a correct prediction and red indicates an incorrect one.}
    \label{fig:insight}
    \vspace{-18pt}
\end{figure}

\subsection{Smoothed Sample Generation and Prediction}
\label{subsec:smoothed_sample_generation_and_prediction}
Different from original samples, a successful AE (adversarial example) is usually crafted elaborately and easily destroyed by perturbation (supported by Figure~\ref{fig:insight}).
A smoothed sample means that a sample is modified around a given original sample while preserving the underlying structure. 
Traditional random smoothing techniques in natural language~\cite{2023-RanMASK, 2020-Safer} employ perturbations on arbitrary tokens of a given code snippet.
However, these approaches compromise the original semantics and syntax of the code, potentially leading to compilation failure.
Figure~\ref{fig:AE_creation} showcases an example.

Given an input code snippet $\vx_0$, \ours{} generates $N$ smoothed samples by perturbing identifiers in two steps: (\circled{1}) perturbation position selection and (\circled{2}) perturbation operation application.
In step (\circled{1}), \ours{} restricts perturbations to identifiers, as modifications to them do not break the code's syntax.
For each code snippet, we randomly select $k$ identifiers. All positions of each selected identifier are then consistently perturbed, and this process is repeated $N$ times to generate a set of smoothed samples, denoted as $s_1, s_2, \cdots, s_N$.
Importantly, position selection is performed at the identifier level. All occurrences of the same identifier are perturbed consistently, preserving code semantics and syntactic validity.
In practice, we randomly select 90\% of the identifiers for each code snippet.

In step (\circled{2}), \ours{} applies a semantic-preserving character-level edit to each selected identifier. 
Specifically, we define three semantic-preserving operations, including \texttt{Insert}, \texttt{Replace}, and \texttt{Delete}.
The operation \texttt{Insert} randomly inserts a new character into the identifier, operation \texttt{Replace} randomly replaces a character in the identifier with a new one, and operation \texttt{Delete} randomly deletes a character from the identifier. 
As all operations are conducted on identifiers, the perturbation is semantic-preserving and will not destroy the code syntax. 
Table~\ref{tab:op_list} illustrates them with corresponding examples.
These operations constitute a fundamental set of character-level edits that can express any transformation between identifier strings. In other words, any adversarial modification to an identifier can be decomposed into a sequence of these basic operations. 
To formally represent such transformations, we introduce the concept of a perturbation path $\mathcal{P}=\{o_1, o_2, \dots, o_L\}$, which denotes an ordered sequence of operations applied to an identifier.
The length of the perturbation path indicates the extent of the identifier to be perturbed.
Following the definition in Eq~\ref{eq:formal_AE}, we can obtain a perturbed sample $\vs$ by applying a path $\mathcal{P}$ for a given $\vx$:
\begin{equation}
    \begin{aligned}
     \vs &= (1-\vm) \odot \vx + \mathcal{P}(\vm \odot \vx) \\
      &=(1-\vm) \odot \vx + \prod_{i=1}^{|\vx|} \prod_{j=1}^{|\mathcal{P}|} o_{j}((\vm \odot \vx)_i), \forall \vx \in \mathcal{X}
\end{aligned}
\label{eq:random_smoothing}
\end{equation}
where $\vm$ is the mask matrix to select identifiers from $\vx$. 
Based on pre-defined operations, \ours{} produces a set of smoothed code snippets (\circled{2}). 
For example, given a code snippet \lstinline[style=mystyle]|int f(void *env) {...}|, \ours{} obtains a smoothed code after operation $Insert$, \lstinline[style=mystyle]|int f(void *env) {...}|. 
\ours{} inserts \texttt{Q} into the identifier $env$.
The insert position and character are randomly picked.
Similarly, \ours{} obtains \lstinline[style=mystyle]|int f(void *enQ){...}| 
and \lstinline[style=mystyle]|int f(void *nv){...}| after operations \texttt{Replace} and \texttt{Delete}.
In practice, we introduce a hyperparameter called perturbation rate $\eta$, which determines the expected proportion of characters to be perturbed within a selected identifier. For each identifier, we adaptively compute a character-level perturbation budget based on its length, resulting in a total of $\eta \cdot |\mathcal{P}|$ edit operations. These operations are then applied randomly to individual characters. 
Thus, the actual perturbation path length $|\mathcal{P}|$, i.e., the number of character-level edits, is indirectly controlled by $\eta$, ensuring that perturbation strength scales with identifier length.

Subsequently, we input these smoothed samples into NCMs to obtain the corresponding predictions $y_1, y_2, \ldots, y_N$, which are then used in the following stages (\circled{3}). 

\subsection{Voting-Based Prediction Aggregation}
To mitigate the influence of incorrectly classifying AEs, \ours{} utilizes a voting-based prediction mechanism to determine the final predictive label.
Rather than predicting based only on a single sample, we generate a collection of smoothed samples and subsequently aggregate the predictions corresponding to this ensemble.
The motivation for this step is that, although one single smoothed sample may not disable an attack, a collection of perturbed samples tends to neutralize such attacks. 
Figure~\ref{fig:insight} supports empirical evidence of the insight.
The x-axis denotes the code index for AEs crafted by different attack methods, and the y-axis denotes the score.
Observe that around a crafted AE, most samples are benign samples and will be predicted as the ground truth label.


In phase (a), we obtain the corresponding predictions $y_1, y_2, \ldots, y_N$. Therefore, \ours{} further aggregates all predictions to determine $\tilde{y}$.
Below is a formal expression:
\begin{equation}
    \tilde{y} = \underset{y \in \mathcal{Y}}{\arg \max }\left[ \sum_{i=1}^{N}\mathbb{I}\{f(\vs_i)=y\} \right],
    \label{eq:formal_majority_voting}
\end{equation}
where $\mathcal{Y}$ is the set of all possible class labels.
$\mathbb{I}$ returns \textbf{1} if the condition inside the parentheses is true and \textbf{0} otherwise.
$\tilde{y}$ is the output prediction.
From the above definition, we can further formally define a smoothed NCM $g(\cdot)$ from a given base NCM $f(\cdot)$: 
\begin{equation}
    g(\vx)=\underset{y \in \mathcal{Y}}{\arg \max }\left[\underset{\vs \sim \mathcal{D}}{\mathbb{P}}(f(\vs)=y)\right],
    \label{eq:smoothed_model}
\end{equation}
where $\mathcal{D}$ is the distribution of the smoothed sample $\vs$.
To simplify, we define:
\begin{equation}
    g(\vx, y)=\underset{\vs \sim \mathcal{D}}{\mathbb{P}}(f(\vs)=y),
\end{equation}
which is the probability (``soft score'') that $f$ returns the class $y$ after perturbation.
Intuitively, $g(\vx)$ denotes the most likely class $y$ to be returned and $g(\vx, y)$ the corresponding probability.

\subsection{Certified Radius Generation}
\label{subsec:certified_radius_generation}

Based on the smoothed predictions, \ours{} generates a certified radius through two sequential steps: bound estimation (\circled{5}) and certified radius determination (\circled{6}). 

\noindent\textbf{Bound Estimation.}
We propose a practical algorithm in \ours{} to estimate lower and upper bounds (\underline{$g(\vx, y)$} and $\overline{g(\vx, y)}$), which is necessary to determine the certified radius.
Beta distribution is generally utilized to model the distribution of probabilities and to estimate the confidence intervals (upper and lower bounds) of a proportion under a confidence level.
The estimation is based on the observed number of successes and failures in a series of trials. 
In that case, we employ it to estimate the upper bound and lower bounds for the probability of success.
Specifically, $N$ copies of a given $\vx$ are crafted in the perturbation step, and then $N$ predictions are conducted.
Specifically, during the perturbation step, $N$ copies of a given $\vx$ are generated, followed by forward propagation through the model.
It can be considered as $N$ independent trials.
According to Eq~\ref{eq:formal_majority_voting}, let $c$ denote the ground truth label.
Then, we define $n_c=\sum_{1}^{N} \mathbb{I}(f(s) = c)$, meaning the number of times $f(\cdot)$ predicts $c$ among $N$ copies for a given $\vx$.
We can assume that $n_c$ follows a binomial distribution with parameters $n$ and $n_c / n$.
That's to say, ($n_{c} \sim \mathrm{~B}\left(n, n_c/n\right)$).
The Beta distribution serves as the conjugate prior to the binomial distribution, allowing us to derive the posterior distribution of the success probability and compute its lower and upper bounds using the quantiles of the Beta distribution.
Then, we have:
\begin{equation}
    \begin{split}
        &\overline{g(\vx, c)} = Beta(1-\alpha, n_c, n-n_c+1) \\
        &\underline{g(\vx, c)} = Beta(\alpha, n_c, n-n_c+1), 
    \end{split}
    \label{eq:beta_for_bounds}
\end{equation}
where $Beta(\alpha, u, v)$ means the $\alpha$-th quantile of a beta distribution with parameters $u$ and $v$.
We set $\alpha$ to 0.001 by default.
Intuitively speaking, we utilize the $\alpha$-th and $(1-\alpha)$-th quantiles of a beta distribution with parameters $n_c$ and $n-n_c+1$ to estimate the lower and upper bounds of $g(\vx, c)$.

\noindent\textbf{Certified Radius Determination.}
Recall that the certified radius indicates a theoretical bound for a given input, within which an adversary can never attack successfully. 
Through Theorem~\ref{theo:bound}, we establish an upper bound on the probability of consistent predictions within an input neighborhood, providing a theoretical basis for Theorem~\ref{theo:certification} (i.e., the consistency between $g(\vx, y)$ and $g(\vx', y)$). Theorem~\ref{theo:certification} then uses this bound to determine a certified radius, reporting certified model robustness against adversarial code snippets within this range. This probability-based approximation enables \ours{} to determine certified robustness without accessing NCM's parameters.
To report the certified robustness of an NCM as discussed in Definition~\ref{def:certified_robustness}, we first declare the distance calculation between $\vx$ and $\vx'$.
The certified radius is determined under this distance.
As adversaries can only modify the identifiers of a given code snippet $\vx$, the difference between $\vx$ and $\vx'$ only exists in identifiers.
Following Eq~\ref{eq:formal_AE}, $\vm \odot \vx$ denotes the extracted identifiers from $\vx$ (i.e., $\{\textbf{0}, iden_0, ..., \textbf{0}\}$), where $iden_0$ means the first identifier and $\textbf{0}$ the masked code token.
In that case, $\Vert\vm\Vert_0=\Vert\vm \odot \vx\Vert_0$ means the number of identifiers of $\vx$.
Let $\Vert\vx - \vx'\Vert_0$ denote the distance between $\vx$ and $\vx'$, it can be calculated as:
\begin{equation*}
    \Vert\vx - \vx'\Vert_0 = \sum_{i=0}^{|\vx|}{\Vert \vx_i - \vx'_i \Vert_0} = \sum_{i=0}^{\Vert\vm\Vert_0}{\Vert(\vm \odot \vx)_i - (\vm \odot \vx')_i\Vert_0}
    \label{eq:distance}
\end{equation*}

To simplify, we assume $\vx$ contains $N$ tokens, where $h$ identifiers exist, and each identifier is considered as an independent special token.
We randomly perturb $k$ identifiers of them.
Let $\vx \ominus \vx'$ denote the set of identifier indices at which $\vx$ and $\vx'$ differ.
In that case, $\vx \ominus \vx'$ = $|\vx \oslash \vx'|$ in Definition~\ref{def:certified_robustness} ($\vx$ and $\vx'$ only differ in identifiers).
For instance, if $\vx \oslash \vx'$ = {1, 3}, the first and third identifiers of $\vx$ and $\vx'$ are different.
To further formalize the smoothed sample $\vs$, we describe the $\vs$ generation process from the view of indices.
Consider $S$ as the set of indices $\{1, 2,..., h\}$.
Let $\mathcal{I}(h, k) \subseteq P(S)$ represent the collection of all subsets of $S$ that contain exactly $k$ distinct elements, where $P(S)$ denotes the power set of $S$. Define $\mathcal{U}(h, k)$ as the uniform distribution over $\mathcal{I}(h, k)$. Sampling from $\mathcal{U}(h, k)$ involves selecting $k$ unique indices from $S$ in such a way that each possible subset has an equal probability of being chosen. For example, a sample from $\mathcal{U}(3, 2)$ could yield the subset $ \{1, 2\}$.
We further define $\mathcal{M}:\mathcal{X}\times\mathcal{I}(h,k) \xrightarrow{} \mathcal{X}'$, where $\mathcal{X}'$ indicates a code snippet in which some identifiers are masked.
$\mathcal{M}$ takes $h$ identifiers and a set of indices as inputs, outputting the masked identifier set.
Then, we can obtain a random smoothed sample.
$\mathcal{U}(h,k)$ is the uniform distribution over $\mathcal{I}(h,k)$. 
The creation of a smoothed sample $\vs$ can be formalized as:
\begin{equation*}
    \vs = \mathcal{M}(\vx,\mathcal{H}) \quad \vx \in \mathcal{X}, \mathcal{H} \sim \mathcal{U} (h_x, k_x)
\end{equation*}
where $\vx$ is the original code snippet, $\mathcal{H}$ is the indices set that is to be perturbed selected from the uniform distribution $\mathcal{U} (h_x, k_x)$.
The $h_x$ denotes the total number of identifiers in $\vx$, while $k_x$ represents the subset of selected identifiers.

\begin{theorem}
\vspace{-4pt}
    Given $\vx$ and $\vx'$, if $|\vx \oslash \vx'| \leq r$, we have:
    \begin{center}
        $g(\vx, y)-g(\vx', y)\leq \beta \times \overline{g(x, y)}$,
    \end{center}
    where $\beta=1-\dfrac{C_{h_x-r}^{k_x}}{C_{h_x}^{k_x}}$ and $\overline{g(\vx, y)}$ the upper bound of $g(\vx, y)$.
    \label{theo:bound}
\end{theorem}
\begin{proof}
Recall that $\mathcal{H}\sim{U}(h_x,k_x)$ and :
\begin{align*}
    g(\vx, y)&= \mathbb{P}(f(\vs)=y) = \mathbb{P}(f(\mathcal{M}(\vx,\mathcal{H}))=y),
    \\
    g(\vx', y)&=\mathbb{P}(f(\vs')=y) = \mathbb{P}(f(\mathcal{M}(\vx',\mathcal{H}))=y)
\end{align*}
Based on the law of total probability:
\begin{align*}
    g(\vx, y) &= \mathbb{P}([f(\mathcal{M}(\vx,\mathcal{H}))=c] \land [\mathcal{H}\cap (\vx\oslash\vx')\neq\emptyset]) \\
    &\quad + \mathbb{P}([f(\mathcal{M}(\vx,\mathcal{H}))=c]\land[\mathcal{H}\cap (\vx\oslash\vx')=\emptyset]), \\
    g(\vx', y) &= \mathbb{P}([f(\mathcal{M}(\vx',\mathcal{H}))=c]\land[\mathcal{H}\cap (\vx\oslash\vx')\neq\emptyset]) \\
    &\quad + \mathbb{P}([f(\mathcal{M}(\vx',\mathcal{H}))=c]\land[\mathcal{H}\cap (\vx\oslash\vx')=\emptyset])
\end{align*}
When $\mathcal{H}\cap (\vx\oslash\vx')=\emptyset$, we have $\mathcal{M}(x,\mathcal{H})=\mathcal{M}(x',\mathcal{H})$.
Intuitively, $\mathcal{H}$  represents the indices of identifiers that are retained (i.e., not masked), while $\vx\oslash\vx'$ denotes the indices at which $\vx$ and $\vx'$ differ.
If their intersection is empty, it implies that all retained identifiers are identical, and all differing identifiers between $\vx$ and $\vx'$ are masked.
Then,

\begin{equation}
    \begin{aligned}
    \mathbb{P}([f(\mathcal{M}(\vx,\mathcal{H}))=c]\land[\mathcal{H}\cap (\vx\oslash\vx')=\emptyset]) \\
    =\mathbb{P}([f(\mathcal{M}(\vx',\mathcal{H}))=c]\land[\mathcal{H}\cap (\vx\oslash\vx')=\emptyset])
    \end{aligned}
\end{equation}

Therefore, 
\begin{align}
    & g(\vx, y) - g(\vx', y) \notag\\
    = & \quad \mathbb{P}([f(\mathcal{M}(\vx,\mathcal{H}))=c]\land[\mathcal{H}\cap (\vx\oslash\vx')\neq\emptyset])\notag\\ 
    &- \mathbb{P}([f(\mathcal{M}(\vx',\mathcal{H}))=c]\land[\mathcal{H}\cap (\vx\oslash\vx')\neq\emptyset]) \notag\\
    \leq & \quad \mathbb{P}([f(\mathcal{M}(\vx,\mathcal{H}))=c]\land[\mathcal{H}\cap (\vx\oslash\vx')\neq\emptyset]) \notag\\
    \leq & \quad \mathbb{P}(f(\mathcal{M}(\vx,\mathcal{H}))=c) \times \mathbb{P}(\mathcal{H}\cap(\vx\oslash\vx')\neq\emptyset) \notag\\
    = & \quad g(\vx, y) \times \mathbb{P}(\mathcal{H}\cap(\vx\oslash\vx')\neq\emptyset) \notag\\
    \leq & \quad \overline{g(\vx, y)} \times \mathbb{P}(\mathcal{H}\cap(\vx\oslash\vx')\neq\emptyset)
    \notag\\
    \label{eq:subtraction}
\end{align}

By using combinatorial methods, we calculate that 
\begin{align}
    \centering
    \mathbb{P}(\mathcal{H}\cap(\vx\oslash\vx')\neq\emptyset)
    &= \quad 1-\mathbb{P}(\mathcal{H}\cap(\vx\oslash\vx')=\emptyset)
    \notag\\
    &= \quad 1 - \dfrac{C_{h_x - |\vx\oslash\vx'|}^{k_x}}{C_{h_x}^{k_x}}
    \notag\\
    &\leq \quad 1 -  \dfrac{C_{h_x-r}^{k_x}}{C_{h_x}^{k_x}} = \beta
    \notag\\
    \label{eq:combination}
\end{align}
Combining Eq~\ref{eq:subtraction} and Eq~\ref{eq:combination}, we have $g(\vx, y)-g(\vx', y)\leq \beta \times \overline{g(\vx, y)}$, thus bounding the prediction under perturbation.
\end{proof}

\begin{theorem}
    Given $\vx$ and $\vx'$, if $\Vert\vx\oslash\vx'\Vert\leq r$ and $\underline{g(\vx, y)}-\beta\times\overline{g(\vx, y)} > 0.5$, with probability at least $(1-\alpha)$, we have:
    $g(\vx') = y$.
    \label{theo:certification}
\end{theorem}

\begin{proof}
    With probability at least $(1-\alpha)$, we have:
    \begin{equation}
        0.5 < \underline{g(\vx, y)} - \beta\times\overline{g(\vx, y)} 
        \leq g(\vx, y) - \beta\times\overline{g(\vx, y)} 
        \leq g(\vx', y)
    \end{equation}
    where the last inequality is from Theorem~\ref{theo:bound}, and $g(\vx')=y$ from its definition given in Eq~\ref{eq:smoothed_model}.
    If $y=c$ (the ground truth label from $\vx$) and $\underline{g(\vx, y)}-\beta\times\overline{g(\vx, y)} > 0.5$, the NCM $g(\cdot)$ is ensured certified robustness around $\vx$.
\end{proof}

The upper and lower bound of $g(\vx, y)$ are estimated previously, and $\beta$ can be calculated by Eq~\ref{eq:beta_for_bounds}.
From Theorem~\ref{theo:certification}, we can report the final certified robustness by a step-by-step practical algorithm to approximate the certified radius.


\begin{algorithm}[t]
    \small
    \caption{Practical Algorithm to Determine Certified Radius}
    \label{alg:practical_algorithm}
    \begin{algorithmic}[1]
        \REQUIRE smoothed predictions $y_1, y_2, \ldots, y_N$, original code snippet $\vx$, ground truth label $c$, the target NCM $g(\cdot)$
        \ENSURE certified robustness radius $r$
        \STATE $\tilde{y} \gets$ get the final prediction from the majority voting.
        \IF{$\tilde{y} \neq c$}
            \RETURN $ABSTAIN$
        \ELSE
             \STATE $\underline{g(\vx, y)},  \overline{g(\vx, y)} \gets$ get the lower/upper bound by Bound Estimation through Eq~\ref{eq:beta_for_bounds}
             \STATE $h_x \gets$ get the number of identifiers in $\vx$
             \FOR{$r$ \textbf{in} $0, \ldots, h_x$} 
                \STATE $\beta \gets$ get $\beta$ value through Eq~\ref{eq:combination}
                \IF{$\underline{g(\vx, y)}-\beta\times\overline{g(\vx, y)} > 0.5$}
                \STATE $r \gets r+1$
                \ELSE
                \STATE break
                \ENDIF
            \ENDFOR
        \ENDIF
    \RETURN $r$ 
  \end{algorithmic}
\end{algorithm}

\noindent\textbf{Practical algorithm.}
Recall that \ours{} operates at the inference stage of NCMs, requiring only model feedback without access to internal model information.
Algorithm~\ref{alg:practical_algorithm} illustrates the practical algorithm to determine the certified radius.
Given the smoothed predictions $y_1, y_2, \ldots, y_N$ and the target NCM $g(\cdot)$, we first determine the final prediction $\tilde{y}$ and assess its correctness (lines 1-3).
Based on Eq~\ref{eq:beta_for_bounds}, then we estimate the lower and upper bounds of $g(\vx, y)$ (line 5).
We incrementally increase $r$ (the number of identifiers to be perturbed) from 0 by 1 (line 7) and calculate $\beta$ using Eq~\ref{eq:combination} (line 8).
This process will continue until $\underline{g(\vx, y)} - \beta \times \overline{g(\vx, y)} \leq 0.5$ (Theorem~\ref{theo:certification}) (line 13).
Upon termination, \ours{} returns $r$ as the certified robustness for $\vx$ (line 16).
Thus, we can report that $g(\vx')$ will return $c$ for any AE $\vx'$ if $|\vx\oslash\vx'| \leq r$ with $(1-\alpha)$ confidence.

\section{Evaluation}
\label{sec:evaluation}

To ensure a thorough evaluation, we conduct comprehensive experiments assessing \ours{} across four distinct aspects: robustness, generalization, time cost, and ablation studies focusing on core components and hyper-parameters. 
We complete the evaluation by answering the following research questions:
\begin{itemize}[leftmargin=25pt]
    \item[\textbf{RQ1.}]\textbf{How robust is \ours{}, from views of empirical and theoretical perspective?}
    \item[\textbf{RQ2.}]\textbf{How efficient is \ours{}?}
    \item[\textbf{RQ3.}]\textbf{How effectively does \ours{} demonstrate generalization across architectures and tasks?}
    \item[\textbf{RQ4.}]\textbf{How do different parameters (including the smoothed sample number $N$ and perturbation rate $\eta$) affect the performance of \ours{}?}
\end{itemize}

\subsection{Experimental Setups}
\label{subsec:setup}
\noindent\textbf{Code understanding tasks and dataset usage.}
Three code understanding tasks and their counterpart datasets are included in our experiments, including defect detection~\cite{zhou2019devign,zhang2022towards}, clone detection~\cite{svajlenko2014towards}, and function classification~\cite{2016-Convolutional-Neural-Networks-over-Tree-Structures}.
Detailed descriptions of the datasets are summarized in our GitHub~\cite{2024-Our-Artifacts}.
\noindent\textbf{Code understanding models usage.}
Four different NCM architectures, CodeBERT~\cite{feng2020codebert}, GraphCodeBERT~\cite{guo2020graphcodebert}, and two code LLMs, CodeLlama~\cite{roziere2023code} and StarCoder~\cite{lozhkov2024starcoder}, are considered. 
Details are available at Github~\cite{2024-Our-Artifacts}.

\subsection{Metrics}
\label{subsec:metrics_learning}

In our experiment, following existing adversarial techniques~\cite{2020-MHM,2022-ALERT}, we utilize \textit{Accuracy (Acc)} to evaluate the effectiveness, \textit{Attack Successful Rate (ASR)} for empirical robustness, respectively.
Given a specific code understanding dataset $\mathcal{X}=\{(\vx_1, c_1), (\vx_2, c_2), \dots, (\vx_n, c_n)\}$, $\vx_i$ denotes a code snippet and $c_i$ is the ground-truth label of $\vx_i$. 
Let $f$ denote the NCM and $f(\vx)$ denote NCM's prediction of code snippet $\vx$.
\textit{ACC} measures the percentage of code snippets that are classified by $f(\cdot)$ correctly and is computed as follows:
\begin{equation}
    ACC = \dfrac{1}{|\mathcal{X}|}\sum_{i=1}^{|\mathcal{X}|}{\mathbb{I}\{f(\vx_i)=c_i\}}, \quad \vx_i \in \mathcal{X}
    \label{eq:ACC}
\end{equation}
where $\mathbb{I}$ returns \textbf{1} when $f(\vx_i)$ = $c_i$ and \textbf{0} otherwise.

Let $\mathcal{X}^*$ denote the corresponding AE dataset based on $\mathcal{X}$.
$\mathcal{X}^* = \{(\vx'_1, c_1), (\vx'_2, c_2), \dots, (\vx'_n, c_n)\}$, $\vx'_i$ is the counterpart AE of original sample $\vx_i$ and $c_i$ the ground-truth label.
ASR measures the percentage of AEs that are misclassified and is defined as follows:
\begin{equation}
    ASR = \dfrac{1}{|\mathcal{X^*}|}\sum_{i=1}^{|\mathcal{X}^*|}{\mathbb{I}\{f(\vx'_i)\neq c_i\}}, \quad \vx'_i \in \mathcal{X^*}
    \label{eq:ASR}
\end{equation}
where $\mathbb{I}$ returns \textbf{1} when $f(\vx_i) \neq c_i$ and \textbf{0} otherwise.


However, ASR is limited to empirical assessments of robustness. To evaluate robustness from a theoretical perspective, we introduce a new metric: the \textit{Normalized Certified Robustness Radius (NCRR)}.
Given an original sample $\vx$, the maximum certified robustness radius denotes the minimal distance from $\vx$ to the decision boundary.
Recall that the distance in a programming language is defined by identifiers, as shown in Eq~\ref{eq:distance}.
Based on the calculated $r$ by Algorithm~\ref{alg:practical_algorithm}, we then define $NCRR$ as follows:
\begin{equation}
    NCRR = \dfrac{1}{|\mathcal{X}|}\sum_{i=1}^{|\mathcal{X}|}{\dfrac{r_x}{h_x}}, \quad
    \vx \in \mathcal{X}
    \label{eq:NCRR}
\end{equation}
where $r_x$ is the certified robustness result returned by \ours{} and $h_x$ denotes the number of identifiers in $\vx$.
The radius $r_x$ is normalized sample-wisely by $\frac{r_x}{h_x}$.
NCRR indicates the average safe radius that can be perturbed for each sample.
Ideally, the NCRR equals \textbf{1}, indicating that adversaries can never find an AE that attacks successfully, even if they manipulate all identifiers of a given code snippet.


\begin{figure*}[t]
\centering
 \subfloat[Distribution of \#Indentifiers]
    {
        \includegraphics[width=0.31\linewidth]{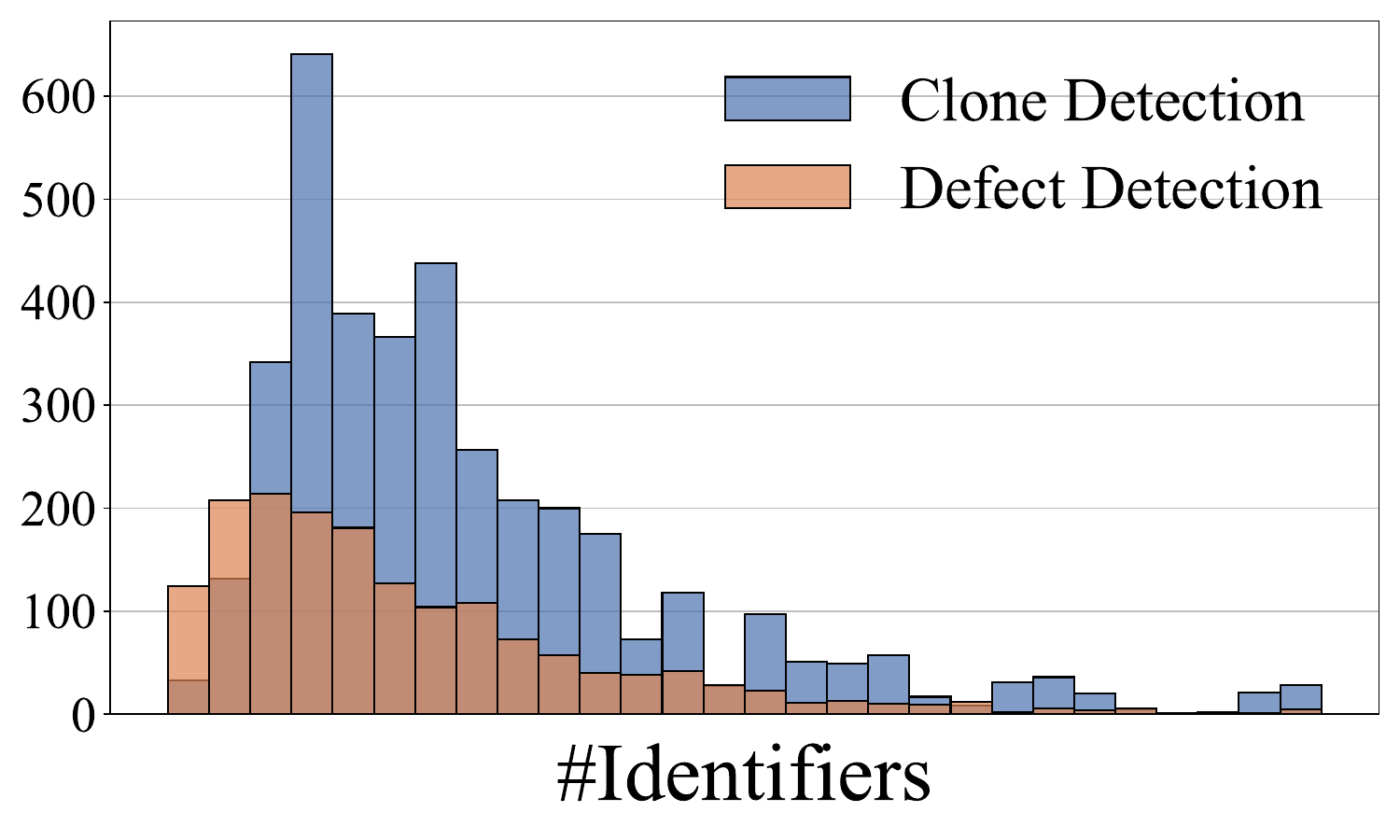}
        \label{fig:distribution}
    }
    \subfloat[Clone Detection (NCRR=0.31)]
    {
        \includegraphics[width=0.31\linewidth]{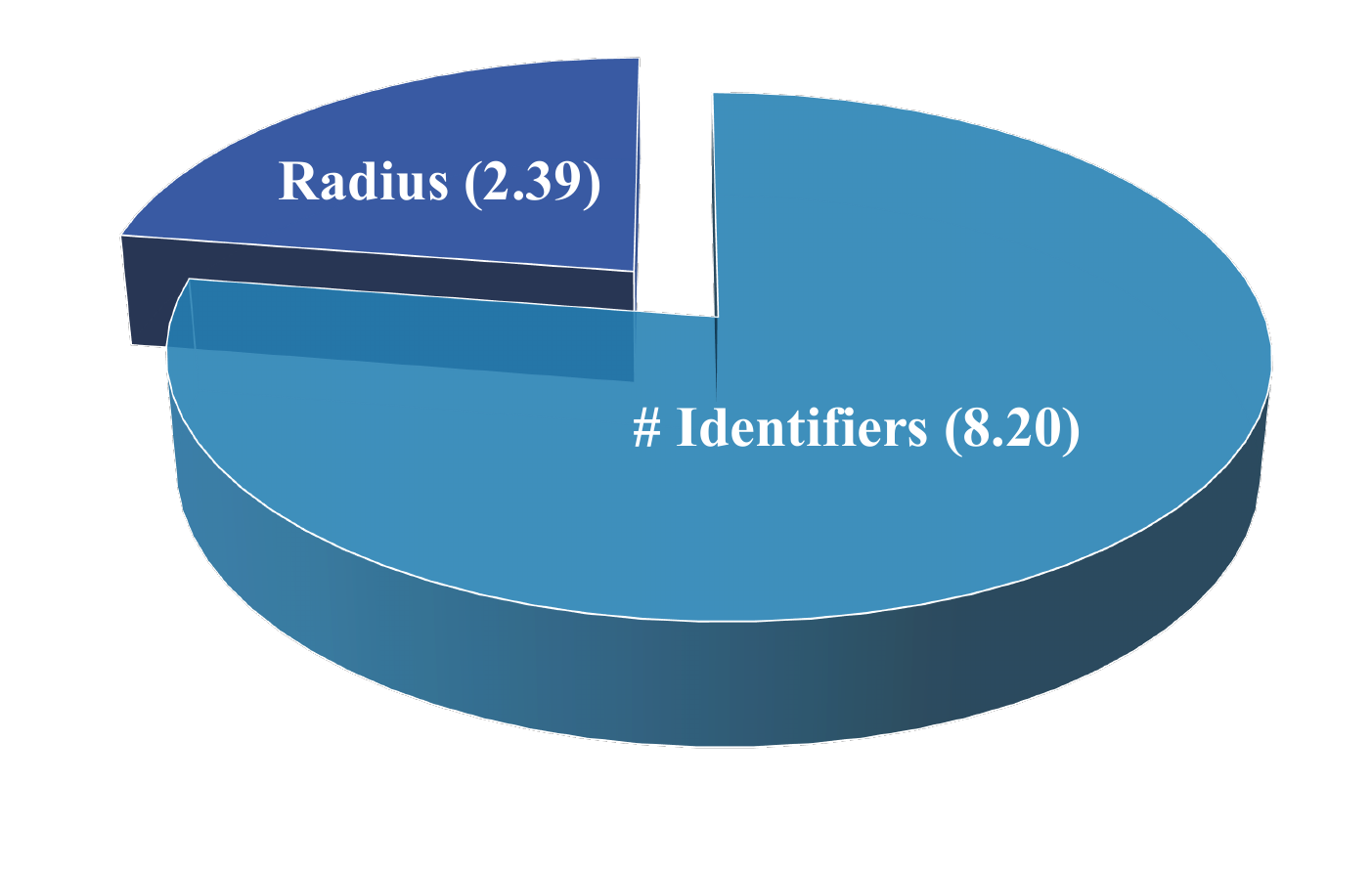}
    }
    \subfloat[Defect Detection (NCRR=0.29)]
    {
        \includegraphics[width=0.31\linewidth]{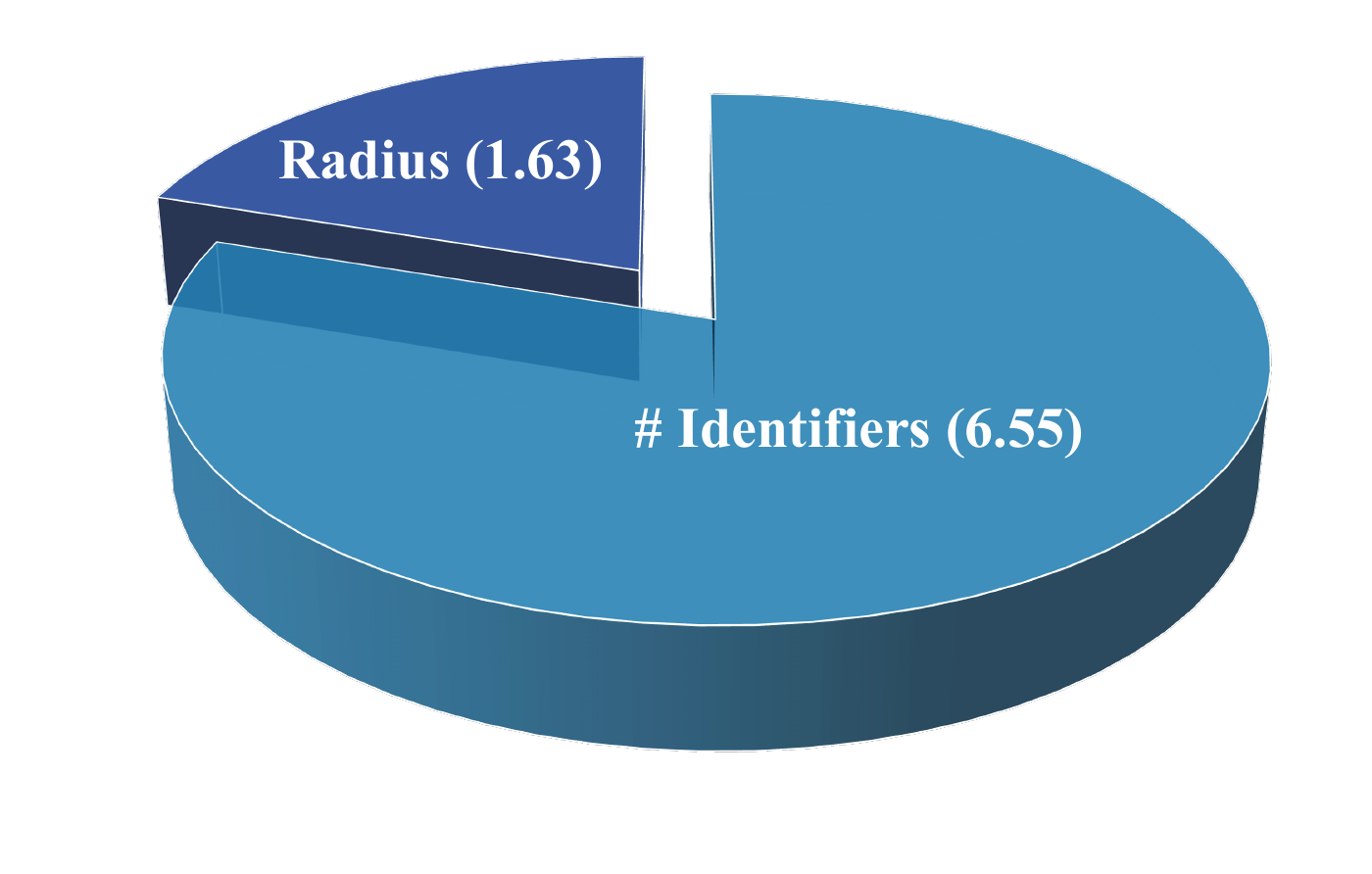}
    }
    \caption{The certified robustness of \ours{}. 
    \#Identifiers mean the number of identifiers for a given code snippet.
    Figure~\ref{fig:certified_robustness}(a) illustrates the distribution of \#Identifiers. Figure~\ref{fig:certified_robustness}(b)(c) report the certified robustness measured by Radius$\uparrow$ and NCRR$\uparrow$.}
    \label{fig:certified_robustness}
    \vspace{-8pt}
\end{figure*}
\begin{table*}[t]
    \centering
    \footnotesize
    \tabcolsep=5pt
    \caption{The empirical robustness (in percentage \%) by different white (\emptycirc) / black (\fullcirc) box defense techniques. Column ``Attack'' denotes different techniques used for constructing AEs. Column ``Victim Model'' denotes the original model without any robustness argumentation techniques.
    ACC is utilized to measure NCM effectiveness and ASR for robustness. ACC$\uparrow$, ASR$\downarrow$.
    }
    \label{tab:main_evaluation}
    \vspace{-6pt}

\begin{tabular}{cccccccccccc}
    \toprule
    Task & Attack  & \multicolumn{2}{c}{Victim Model} & \multicolumn{2}{c}{RoPGen} & \multicolumn{2}{c}{SPACE} & \multicolumn{2}{c}{\ours{} [$N$=100]} &\multicolumn{2}{c}{\ours{} [$N$=1000]}\\
    \midrule
    \multicolumn{2}{c}{White/Black Box} &\multicolumn{2}{c}{-} & \multicolumn{2}{c}{\emptycirc}  & \multicolumn{2}{c}{\emptycirc} &\multicolumn{2}{c}{\fullcirc} &\multicolumn{2}{c}{\fullcirc}\\
    \midrule
    \multicolumn{2}{c}{Metrics} & ACC & ASR & ACC & ASR & ACC & ASR & ACC & ASR & ACC & ASR\\ 
    \midrule
    \multirow{4}{*}{Defect Detection} 
    & MHM & 61.86 & 36.12 & 63.40 & 20.50 & 64.23 & 23.40 & 62.15 & 8.80 &62.18 &8.38\\
    & ALERT & 61.86 & 46.80 & 63.40 & 21.47 & 64.23 & 29.85 & 62.15 & 8.74 &62.18 &8.62\\
    & CODA& 61.86 & 44.39 & 63.40 & 23.34 & 64.23 & 28.64 & 62.15 & 11.70 &62.18 &13.32\\
    & Average & 61.86 & 42.43 & 63.40 & 21.77 & 64.23 & 27.29 & 62.15 & 9.74 &62.18 &10.10\\ 
    \midrule
    \multirow{4}{*}{Clone Detection}  
    & MHM & 96.72 & 8.71 & 97.30 & 2.19 & 97.05 & 5.76 & 96.72 & 0.95 &96.62 &1.08\\
    & ALERT & 96.72 & 17.73 & 97.30 & 3.92 & 97.05 & 10.64 & 96.72 & 1.36 &96.62 &1.42\\
    & CODA & 96.72 & 38.10 & 97.30 & 10.51 & 97.05 & 21.03 & 96.72 & 1.68 &96.62 &1.52\\
    & Average & 96.72 & 21.51 & 97.30 & 3.52 & 97.05 & 12.47 & 96.72 & 1.33 &96.62 &1.34\\ 
    \bottomrule
    \end{tabular}
    \vspace{-3mm}
\end{table*}
\begin{table}[t]
    \footnotesize
        \tabcolsep=6pt
        \centering
        \caption{The efficiency of \ours{} and other methods.  
        The $\pm$ denotes the standard deviation calculated over 1000 repeated inferences on the test set.
        }
        \vspace{-6pt}
        \label{tab:efficiency}
        \begin{tabular}{cccc}
        \toprule
        Task & Method & \begin{tabular}[c]{@{}c@{}}Preparation\\ time cost (s)\end{tabular} & \begin{tabular}[c]{@{}c@{}}Inference\\ time cost (s / sample)\end{tabular} \\
        \midrule
        \multirow{3}{*}{Defect detection} & SPACE & 2,240 & 0.015 $\pm$ 0.039 \\
        & RoPGen & 141,010 & 0.010 $\pm$ 0.034 \\
        & \ours{} & 1,328 & 0.724 $\pm$ 0.122 \\
        \midrule
        \multirow{3}{*}{Clone detection} & SPACE & 10,268 & 0.016 $\pm$ 0.037 \\
        & RoPGen & 342,545 & 0.014 $\pm$ 0.035 \\
        & \ours{} & 1,365 & 1.535 $\pm$ 0.062 \\
        \bottomrule
    \end{tabular}
    \vspace{-14pt}
\end{table}

To summarize, \textit{ACC} measures effectiveness, with higher values indicating greater effectiveness. 
\textit{ASR} measures empirical robustness; lower \textit{ASR} indicates more robust NCMs are.

\subsection{RQ1. The robustness of \ours{}.}
\label{subsec:experimental_results}
\noindent\textbf{How robust is \ours{} empirically?}
\newline
\noindent\textit{Design.}
We evaluate \ours{} on two code understanding tasks: defect detection and clone detection. Table~\ref{tab:main_evaluation} presents results against three adversarial attacks: MHM~\cite{2020-MHM}, ALERT~\cite{2022-ALERT}, and CODA~\cite{2023-CODA}. 
Note that in addition to identifier-targeted attacks (MHM and ALERT), structure-targeted attack CODA is also considered.
Column ``Victim Model'' denotes undefended NCMs, while subsequent columns denote compared current defenses, including RoPGen~\cite{2022-ropgen}, SPACE~\cite{2022-SPACE}, and \ours{} with $N$=100 and $N$=1000 smoothed samples.
All models are evaluated on the same datasets~\cite{2024-Our-Artifacts} and model CodeBERT~\cite{feng2020codebert}. For a fair comparison, we follow RoPGen and SPACE’s original post-training settings. 

\textit{Results and analysis.}
For defect detection, \ours{}[N=100] reduces the average ASR from 42.43\% (Victim Model) to 9.74\%, significantly outperforming RoPGen (21.77\%) and SPACE (27.29\%). Similarly, for clone detection, \ours{} achieves an ASR between 0.95\% and 1.68\%, outperforming RoPGen (3.52\%) and SPACE (12.47\%), while preserving high ACC (>96\%).
While RoPGen and SPACE show marginally higher ACC, their improvements stem from extensive post-training modifications. In contrast, \ours{} operates efficiently at inference time without training phase, making it more practical and adaptable for real-world applications.
Observe that \ours{}[$N$=1000] only achieves results comparable to \ours{}[$N$=100], indicating diminishing returns with increased smoothing samples. Given the trade-off between efficiency and robustness, we adopt $N=100$ as the default configuration.
We further analyze the impact of the parameter $N$ in RQ4 through a detailed ablation study.


In summary, \ours{} consistently achieves the best empirical robustness (lowest ASR) across all tasks and attacks on average while maintaining high effectiveness (ACC), highlighting its superior performance compared to baselines.

\noindent\textbf{How robust is \ours{} theoretically?}
\newline
\noindent\textit{Design.}
We evaluate the certified robustness of \ours{} by computing the certified radius using Algorithm~\ref{alg:practical_algorithm}, which guarantees that no adversarial example can succeed within the perturbation bound $r$ (as defined in Definition~\ref{def:certified_robustness}). The metric ``Radius'' denotes the average number of certified identifiers per task, while ``NCRR'' (Eq.~\ref{eq:NCRR}) represents the normalized percentage of these identifiers. All calculations are conducted on test set for each code understanding task.

\textit{Results and analysis.}
Figure~\ref{fig:certified_robustness} illustrates the reported certified robustness measured by ``Radius'' and ``NCRR''.
``\#Identifiers'' means the average number of identifiers contained in each code snippet.
\ours{} achieves an average radius of 1.63 for defect detection, with each sample containing an average of 6.55 identifiers and an NCRR of 29\%, indicating a substantial proportion of robustly certified identifiers. For clone detection, the average radius is 2.39, with an average of 8.20 identifiers per sample and an NCRR of 31\%, demonstrating \ours{}'s strong robustness boundary against adversarial examples. 

In summary, the results show that \ours{} consistently achieves strong certified robustness across all tasks. A substantial proportion of identifiers are provably robust, highlighting the effectiveness of \ours{} in defending adversarial attacks within the certified perturbation bound.

\begin{figure*}[!t]
    \centering
    \begin{minipage}{0.6\textwidth}
        \includegraphics[width=\linewidth]{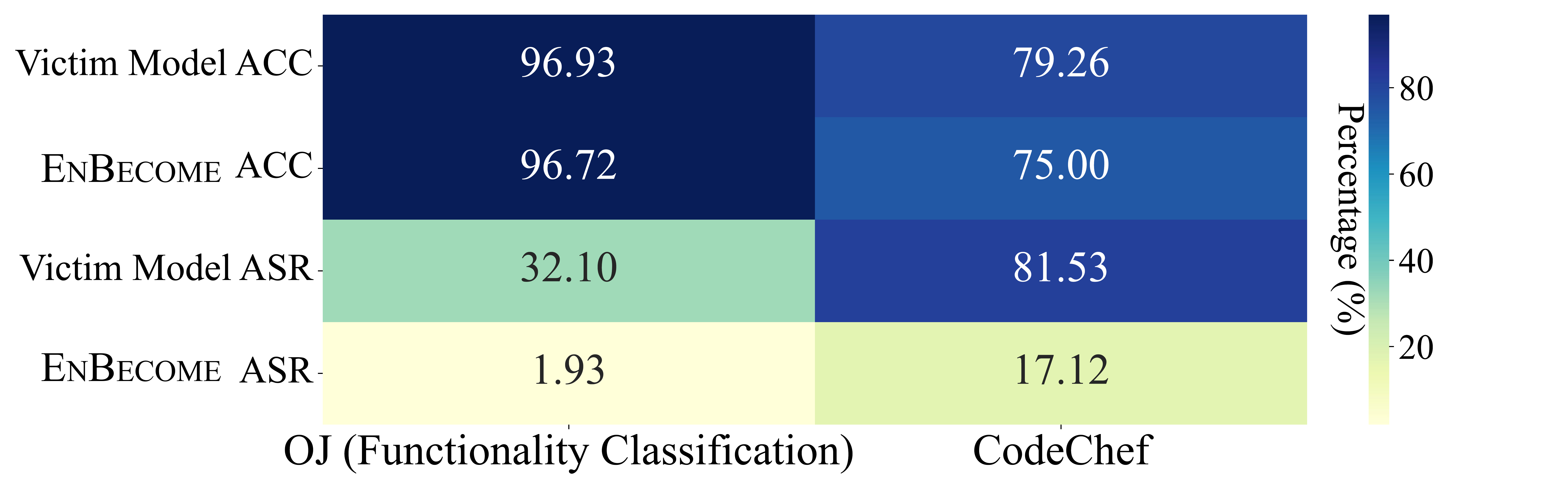} 
        \vspace{-4pt}
        \caption{\ours{} on other datasets, OJ (Functionality Classification) and CodeChef. Darker shades represent higher percentages, with notable differences between the NCMs' ASR on each dataset.}
        \label{fig:gen_other_dataset}
    \end{minipage}
    \hfill
    \begin{minipage}{0.37\textwidth}
        \includegraphics[width=\linewidth]{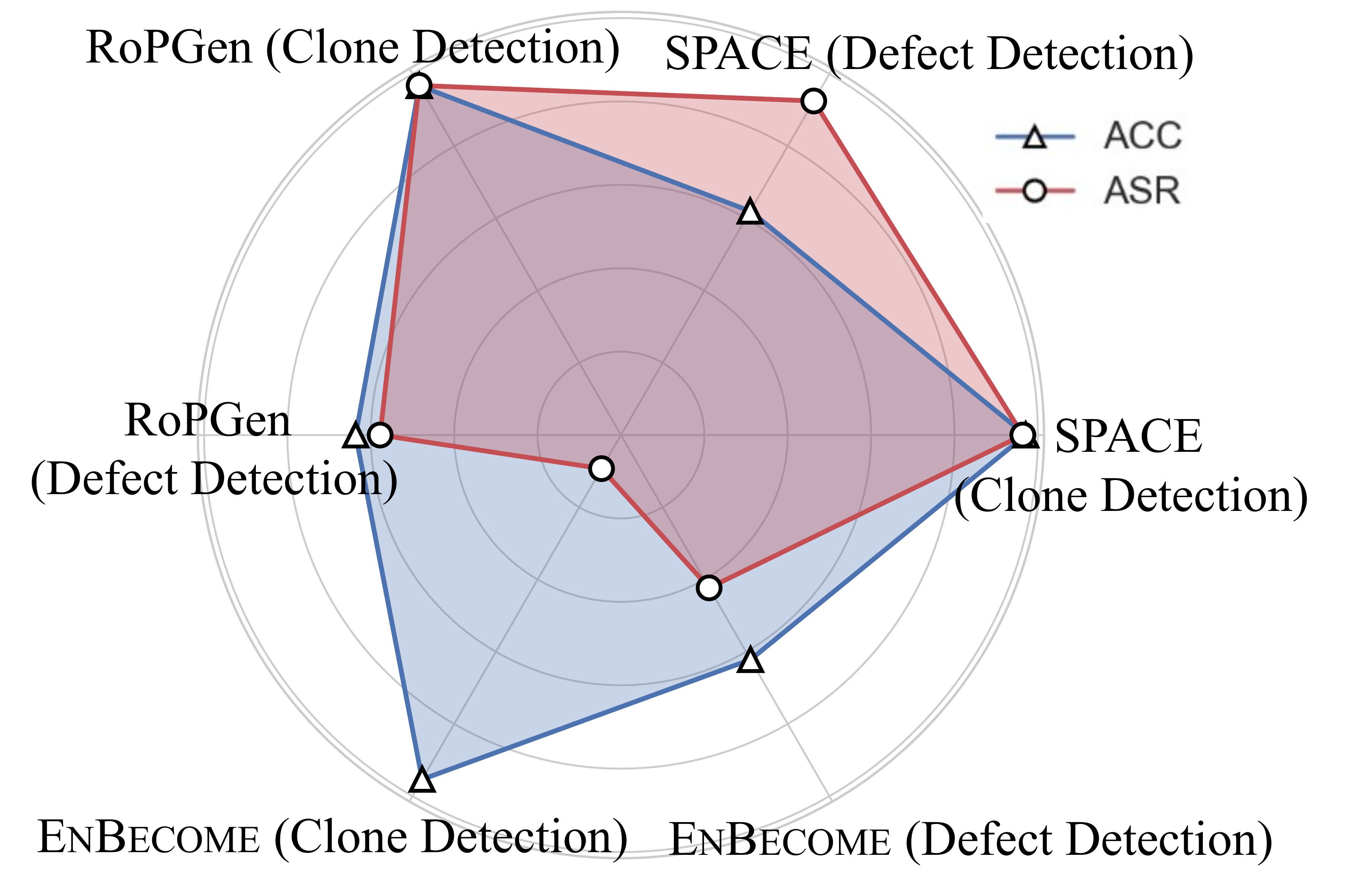} 
        \caption{\ours{} on backdoor attacks.}
        \label{fig:gen_backdoor_attack}
    \end{minipage}
    \vspace{-8pt}
\end{figure*}

\subsection{RQ2. How efficient is \ours{}?}
\label{subsec:time_cost}
\noindent\textit{Design.}
In this section, we compare \ours{} with baseline methods (SPACE and RoPGen) in terms of efficiency, measured through preparation and inference costs. Preparation cost refers to the extra time required before inference: for SPACE and RoPGen, this involves post-training time, while for \ours{}, it consists of generating smoothed samples for the test set.
To ensure a fair comparison, we adopt the default hyperparameters of all baselines and evaluate them on the same test set, which we construct by randomly selecting 1,000 original clean samples. 

\textit{Results and analysis.}
Table~\ref{tab:efficiency} presents the efficiency comparison between \ours{} and other methods, specifically SPACE and RoPGen.
Observe that \ours{} achieves up to 139,682 seconds (\textasciitilde 38 hours) and 341,180 seconds (\textasciitilde 95 hours) of savings compared to RoPGen on clone detection.
While \ours{}’s inference time per sample (0.724s for defect detection, 1.535s for clone detection) is higher than SPACE and RoPGen (0.010\textasciitilde0.016s), this trade-off is acceptable given its training-free, black-box nature and the significant efficiency gain in deployment.
Although \ours{} introduces a slight increase in inference overhead, the average time cost per sample remains under 2 seconds, which is acceptable for practical use.

\subsection{RQ3. The generalization of \ours{}.}
To evaluate \ours{}'s generalization, three different perspectives, including NCM architectures, other code intelligent datasets, and backdoor attacks, are taken into consideration.
\label{subsec:generalization}


\begin{table}[t]
     \tabcolsep=6pt
     \footnotesize
         \renewcommand{\arraystretch}{1.1} 
         \centering
        \caption{The performance of \ours{} on multiple NCM architectures, including traditional architectures (top 3) and code LLMs (bottom 2).}
        \vspace{-6pt}
        \label{tab:model_arch}
        \begin{tabular}{ccccccc}
        \toprule
        \multirow{2}{*}{NCM Arch} & \multicolumn{2}{c}{Victim Model} & \multicolumn{4}{c}{\ours{}} \\
        \cmidrule(lr){2-3}  \cmidrule(lr){4-7} 
         & ACC & ASR & ACC & ASR & Radius & NCRR\\
         \midrule
        GraphCodeBERT & 64.31 & 79.11 & 63.10 & 25.46 & 1.43 & 0.17\\
        CodeBERT & 61.86 & 46.80 & 62.15 & 8.74 & 1.85 & 0.22\\
        CodeT5 & 64.20 & 77.88 & 63.61 & 13.32 & 1.61 & 0.22\\
        \midrule
         StarCoder & 50.14 & 86.54 & 49.01 & 30.04 & 1.74 & 0.24\\
        CodeLlama & 48.16 & 68.74 & 45.75 & 28.65 & 1.02 & 0.13\\
        \bottomrule
        \end{tabular}
    \vspace{-12pt}
\end{table}

\begin{table}[!t]
    \centering
    \footnotesize
     \setlength{\abovecaptionskip}{2pt}  
    \setlength{\belowcaptionskip}{2pt}
    \caption{Ablation study.
    Parameters $N$ and $\eta$ of \ours{} are considered. $N$ denotes the smoothed sample count and $\eta$ denotes the perturbation rate.}
    \label{tab:ablation_study}
    \tabcolsep=5pt
\begin{tabular}{lcccccccc}
\toprule
\multirow{3}{*}{Param: $N$} & \multicolumn{2}{c}{$N$=0} & \multicolumn{2}{c}{$N$=100} & \multicolumn{2}{c}{$N$=1000} & \multicolumn{2}{c}{$N$=10000} \\
\cmidrule(lr){2-3} \cmidrule(lr){4-5} \cmidrule(lr){6-7} \cmidrule(lr){8-9}
 & ACC & ASR & ACC & ASR & ACC & ASR & ACC & ASR \\
 \cmidrule(lr){2-9}
 & 64.3 & 47.0 & 62.9 & 8.8 & 63.3 & 8.9 & 63.2 & 8.9 \\
 \midrule
\multirow{3}{*}{Param: $\eta$} & \multicolumn{2}{c}{$\eta$=0} & \multicolumn{2}{c}{$\eta$=0.3} & \multicolumn{2}{c}{$\eta$=0.6} & \multicolumn{2}{c}{$\eta$=1.0} \\
\cmidrule(lr){2-3} \cmidrule(lr){4-5} \cmidrule(lr){6-7} \cmidrule(lr){8-9}
 & ACC & ASR & ACC & ASR & ACC & ASR & ACC & ASR \\
 \cmidrule(lr){2-9}
 & 64.1 & 47.0 & 63.1 & 10.2 & 63.1 & 9.4 & 62.9 & 9.1 \\
 \bottomrule
\end{tabular}
\vspace{-12pt}
\end{table}

\noindent\textbf{\ours{} on multiple model architectures}
\newline
\noindent\textit{Design.}
We evaluate \ours{} on five NCM architectures: GraphCodeBERT~\cite{guo2020graphcodebert}, CodeBERT~\cite{feng2020codebert}, CodeT5~\cite{2021-CodeT5}, and two code LLMs, StarCoder~\cite{2023-starcoderplus} and CodeLlama~\cite{roziere2023code}. The best attack as shown in Table~\ref{tab:main_evaluation}, ALERT~\cite{2022-ALERT}, is used to generate adversarial examples on the task of defect detection.

\textit{Results and analysis.}
Table~\ref{tab:model_arch} summarizes \ours{}'s performance across NCM architectures, evaluating on four different metrics, ACC, ASR, Radius, and NCRR. \ours{} significantly lowers ASR with minimal ACC impact. ASR drops from 79.11\% to 25.46\% (GraphCodeBERT), 46.80\% to 8.74\% (CodeBERT), and 77.88\% to 13.32\% (CodeT5), achieving strong Radius/NCRR values. Code LLMs show similar trends, with StarCoder's ASR reducing from 86.54\% to 30.04\% (Radius: 1.74, NCRR: 0.24) and CodeLlama’s from 68.74\% to 28.65\% (Radius: 1.02, NCRR: 0.13), though with slightly lower ACC due to their black-box nature. CodeLlama’s broader language support makes it more vulnerable than StarCoder. Overall, \ours{} effectively enhances robustness across different NCM architectures, substantially reducing ASR while maintaining model performance.

\noindent\textbf{\ours{} on other datasets.}

\noindent\textit{Design.}
We evaluate \ours{}'s generalization on OJ dataset~\cite{2016-Convolutional-Neural-Networks-over-Tree-Structures} for functionality classification and CodeChef~\cite{zhang2022towards} for defect detection. OJ dataset\cite{codexglue} is based on the Open Judge benchmark, while CodeChef categorizes execution results as ``OK'' (no defect), ``WA,'' ``TLE,'' and ``RE'' (defects).

\textit{Results and analysis.}
Figure~\ref{fig:gen_other_dataset} shows \ours{} on OJ and CodeChef. \ours{} reduces ASR from 32.10\% to 1.93\% on OJ (96.72\% ACC) and from 81.53\% to 17.12\% on CodeChef with only a 4.26\% ACC drop. These results confirm \ours{}'s robustness in enhancing NCMs across diverse intelligent code understanding tasks.

\noindent\textbf{\ours{} on backdoor attacks.}

\noindent\textit{Design.}
To evaluate \ours{}'s performance against various security threats, we consider backdoor attacks, specifically BadCode~\cite{2023-BadCode}, a state-of-the-art method. Unlike traditional adversarial attacks, which induce misclassification via AEs, BadCode implants backdoors that trigger misclassification when specific patterns appear.

\textit{Results and analysis.}
Figure~\ref{fig:gen_backdoor_attack} shows \ours{}'s performance against backdoor attacks in defect and clone detection. In the radar chart, ACC (blue) and ASR (red) are compared across RoPGen, SPACE, and \ours{}. For clone detection, \ours{} reduces ASR from over 96\% to 9.30\%, with minimal ACC drop (95.87\% to 95.35\%). For defect detection, it lowers ASR from 91.55\% to 42.31\% while maintaining 62.18\% ACC, highlighting strong defense capability.

\subsection{RQ4. Ablation Study.}
\label{subsec:ablation_study}
We analyze how \ours{}'s performance is affected by the smoothed sample count $N$ and perturbation rate $\eta$.
All evaluations utilize 1000 defect detection samples with adversarial examples generated by ALERT~\cite{2022-ALERT}.

\noindent\textbf{Influence of the smoothed sample number $N$.}
As outlined in Section~\ref{sec:meth}, \ours{} generates $N$ perturbed samples per input. We evaluate its impact with $N = 100, 1000, 10000$, and revert to the victim NCM when $N = 0$. 
Table~\ref{tab:ablation_study} presents the results. Without perturbations, the victim NCM achieves 64.3\% ACC but suffers from a high ASR of 47.0\%, revealing its vulnerability. With $N = 100$, ACC slightly drops to 62.9\%, but ASR significantly improves to 8.8\%, demonstrating increased robustness with minimal accuracy loss. Beyond $N = 1000$, ACC and ASR stabilize, indicating no further gains in robustness. For efficiency, we set $N = 100$.

\noindent\textbf{Influence of the perturbation rate $\eta$.}
\ours{} perturbs identifiers at a rate of $\eta$ per smoothed sample. We evaluate its impact by testing $\eta = 0.3, 0.6, 1.0$, with $\eta = 0$ reverting to the baseline NCM.  
Table~\ref{tab:ablation_study} presents the results. Without perturbation ($\eta=0$), ASR is 47.0\%, showing high vulnerability. Increasing $\eta$ to 0.3 sharply reduces ASR to 10.2\% with minimal ACC impact. Further increases to 0.6 and 1.0 lower ASR to 9.4\% and 9.1\%, respectively, with stable ACC. Since improvements plateau at $\eta=1.0$, moderate perturbation rates offer the best balance between robustness and accuracy.

\section{Related Works}
\label{sec:related_works}
In this section, we introduce related works of this paper.

\noindent\textbf{Neural Code Models.} 
Neural Code Models (NCMs) leverage deep learning for code-related tasks~\cite{palacio2024toward}, excelling in areas like code search~\cite{gu2018deep, sachdev2018retrieval, 2022-Code-Search-based-on-Context-aware-Code-Translation}, summarization~\cite{leclair2019neural, haque2020improved, ahmad2020transformer, 2024-ESALE, 2024-LLM4CodeSum}, and defect detection~\cite{wu2023large, zhou2019devign}. Traditionally, they treat code as text~\cite{mastropaolo2021studying}, but some incorporate structural representations (e.g., AST, CFG, DFG)~\cite{2019-ASTNN, wan2019multi, zeng2023degraphcs}.  
Advances in transfer learning with models like BERT~\cite{kenton2019bert} and T5~\cite{2020-T5} have driven NCM development. Pre-trained models such as CodeBERT~\cite{feng2020codebert} and GraphCodeBERT~\cite{guo2020graphcodebert} enhance understanding through token prediction and graph-based relationships. Recently, large language models (LLMs) specialized for code, such as StarCoder~\cite{2023-StarCoder} and CodeLlama~\cite{roziere2023code}, have redefined NCM capabilities. 
We focus on NCMs, including code LLMs, for key software engineering code understanding tasks like functionality classification~\cite{2016-Convolutional-Neural-Networks-over-Tree-Structures, 2019-ASTNN}, defect detection~\cite{wu2023large, zhou2019devign}, and clone detection~\cite{2020-Functional-Code-Clone-Detection, svajlenko2014towards, 2023-Literature-Review-Code-Similarity, 2024-AdaCCD}, categorized by their specific functions.

\noindent\textbf{Adversarial Robustness of Neural Code Models.} 
Adversarial robustness refers to a model's ability to maintain correct predictions under adversarial perturbations~\cite{2022-li-adv-semi, 2020-li-adv-targeted}. 
Recent work has adapted NLP adversarial techniques to neural code models (NCMs), developing both attacks and defenses. 
Similarly, the adversarial robustness of NCMs refers to resisting perturbations, such as alerting identifiers, without significantly degrading performance. 
The research into the adversarial robustness of NCMs focuses on exploiting the NCMs used in code intelligence tasks~\cite{2020-MHM, 2022-ALERT, 2023-CODA}. 
For instance, Yang et al.~\cite{2022-ALERT} use greedy search and genetic algorithms to generate natural adversarial samples, while Zhang et al.~\cite{2020-MHM} employ Metropolis-Hastings sampling for identifier renaming. Other works leverage gradient-based optimization~\cite{2023-discrete-adv, 2020-strata, zhang2022towards} or apply code transformations to enhance generalization~\cite{2022-ropgen}. Some also design adversarial training frameworks tailored to code-specific constraints~\cite{2022-SPACE}. However, these methods require post-training, resulting in high computational cost.
In contrast, we introduce \ours{}, a black-box and training-free approach that strengthens NCM robustness without modifying model parameters while also reporting a theoretical bound.

\section{Threats to validate}
\label{sec:threat_to_validate}
In this section, we discuss the potential threats to the validity of \ours{}, considering internal and external factors.
\newline
\noindent\textbf{Internal Validity.}
Internal validity involves two aspects: parameter sensitivity and accuracy reduction. The effectiveness of \ours{} depends on hyper-parameters such as the number of smoothed samples $N$ and perturbation rate $\eta$, whose impact is examined through ablation studies with standard settings. Although estimating the certified radius involves theoretical approximations that may introduce minor inaccuracies, our results show that \ours{} maintains competitive and acceptable accuracy in practice.
\newline
\noindent\textbf{External Validity.}
The external validity threats focus on the generalizability across different NCM architectures. Although we have evaluated \ours{} on several NCM architectures, including CodeBERT, CodeT5, GraphCodeBERT, StarCoder, and CodeLlama, the findings may not generalize to all NCMs for code understanding, especially those with novel architectures or non-standard training regimes.
To address this issue, \ours{} operates independently of the internal parameters of NCMs, relying only on model output predictions. 
This design maximizes the generalizability of \ours{} across different NCMs as a black-box adaptable framework.

By addressing these threats, we aim to reinforce the reliability of our findings and ensure a robust evaluation of \ours{} for NCM defense in code understanding tasks.
\section{Conclusion}
\label{sec:conclusion}


We propose \ours{}, a training-free, black-box method that enhances and reports the robustness boundary of NCMs for intelligent code understanding tasks. Operating at inference time, \ours{} significantly reduces ASR with minimal performance loss and requires no post-training. It is the first to provide certified robustness for NCMs in intelligent code understanding tasks. Extensive evaluations show that \ours{} effectively mitigates adversarial threats and generalizes well across diverse NCM architectures.

\bibliographystyle{IEEEtran}
\bibliography{sample-base}

@String{Computer = "{IEEE} Computer" }

@String{Springer = "Springer-Verlag" }

@article{shin2010evaluating,
  title={Evaluating complexity, code churn, and developer activity metrics as indicators of software vulnerabilities},
  author={Shin, Yonghee and Meneely, Andrew and Williams, Laurie and Osborne, Jason A},
  journal={IEEE transactions on software engineering},
  volume={37},
  number={6},
  pages={772--787},
  year={2010},
  publisher={IEEE}
}

@inproceedings{bielik2020adversarial,
	address = {Virtual Event},
	author = {Pavol Bielik and Martin T. Vechev},
	booktitle = {Proceedings of the 37th International Conference on Machine Learning},
	month = {13-18 July},
	pages = {896--907},
	publisher = {{PMLR}},
	title = {Adversarial Robustness for Code},
	volume = {119},
	year = {2020}}

@inproceedings{chen2017adversarial,
	address = {Athens, Greece},
	author = {Lingwei Chen and Yanfang Ye and Thirimachos Bourlai},
	booktitle = {Proceedings of the European Intelligence and Security Informatics Conference},
	month = {September 11-13},
	pages = {99--106},
	publisher = {{IEEE} Computer Society},
	title = {Adversarial Machine Learning in Malware Detection: Arms Race between Evasion Attack and Defense},
	year = {2017}}

@inproceedings{jha2023codeattack,
  title={Codeattack: Code-based adversarial attacks for pre-trained programming language models},
  author={Jha, Akshita and Reddy, Chandan K},
  booktitle={Proceedings of the AAAI Conference on Artificial Intelligence},
  volume={37},
  number={12},
  pages={14892--14900},
  year={2023}
}

@article{liu2021practical,
  title={A practical black-box attack on source code authorship identification classifiers},
  author={Liu, Qianjun and Ji, Shouling and Liu, Changchang and Wu, Chunming},
  journal={IEEE Transactions on Information Forensics and Security},
  volume={16},
  pages={3620--3633},
  year={2021},
  publisher={IEEE}
}

@inproceedings{2021-CodeT5,
	address = {Virtual Event / Punta Cana, Dominican Republic},
	author = {Yue Wang and Weishi Wang and Shafiq R. Joty and Steven C. H. Hoi},
	booktitle = {Proceedings of the 26th Conference on Empirical Methods in Natural Language Processing},
	month = {7-11 November},
	pages = {8696--8708},
	publisher = {Association for Computational Linguistics},
	title = {CodeT5: Identifier-aware Unified Pre-trained Encoder-Decoder Models for Code Understanding and Generation},
	year = {2021}}

@inproceedings{quiring2019misleading,
  title={Misleading authorship attribution of source code using adversarial learning},
  author={Quiring, Erwin and Maier, Alwin and Rieck, Konrad},
  booktitle={28th USENIX Security Symposium (USENIX Security 19)},
  pages={479--496},
  year={2019}
}

@inproceedings{severi2021explanation,
  title={$\{$Explanation-Guided$\}$ backdoor poisoning attacks against malware classifiers},
  author={Severi, Giorgio and Meyer, Jim and Coull, Scott and Oprea, Alina},
  booktitle={30th USENIX security symposium (USENIX security 21)},
  pages={1487--1504},
  year={2021}
}

@inproceedings{2022-ropgen,
  title={Ropgen: Towards robust code authorship attribution via automatic coding style transformation},
  author={Li, Zhen and Chen, Guenevere and Chen, Chen and Zou, Yayi and Xu, Shouhuai},
  booktitle={Proceedings of the 44th International Conference on Software Engineering},
  pages={1906--1918},
  year={2022}
}

@article{2022-SPACE,
  title={Semantic-preserving adversarial code comprehension},
  author={Li, Yiyang and Wu, Hongqiu and Zhao, Hai},
  journal={arXiv preprint arXiv:2209.05130},
  year={2022}
}

@inproceedings{zhou2019devign,
	address = {Vancouver, BC, Canada},
	author = {Yaqin Zhou and Shangqing Liu and Jing Kai Siow and Xiaoning Du and Yang Liu},
	booktitle = {Proceedings of the 33rd Annual Conference on Neural Information Processing Systems},
	month = {December 8-14},
	pages = {10197--10207},
	publisher = {Curran Associates Inc.},
	title = {Devign: Effective Vulnerability Identification by Learning Comprehensive Program Semantics via Graph Neural Networks},
	year = {2019}}

@article{codexglue,
  title={Codexglue: A machine learning benchmark dataset for code understanding and generation},
  author={Lu, Shuai and Guo, Daya and Ren, Shuo and Huang, Junjie and Svyatkovskiy, Alexey and Blanco, Ambrosio and Clement, Colin and Drain, Dawn and Jiang, Daxin and Tang, Duyu and others},
  journal={arXiv preprint arXiv:2102.04664},
  year={2021}
}

@inproceedings{svajlenko2014towards,
  title={Towards a big data curated benchmark of inter-project code clones},
  author={Svajlenko, Jeffrey and Islam, Judith F and Keivanloo, Iman and Roy, Chanchal K and Mia, Mohammad Mamun},
  booktitle={2014 IEEE International Conference on Software Maintenance and Evolution},
  pages={476--480},
  year={2014},
  organization={IEEE}
}

@article{2022-li-adv-semi,
  title={Semi-supervised robust training with generalized perturbed neighborhood},
  author={Li, Yiming and Wu, Baoyuan and Feng, Yan and Fan, Yanbo and Jiang, Yong and Li, Zhifeng and Xia, Shu-Tao},
  journal={Pattern Recognition},
  volume={124},
  pages={108472},
  year={2022},
  publisher={Elsevier}
}

@article{2023-discrete-adv,
  title={Discrete adversarial attack to models of code},
  author={Gao, Fengjuan and Wang, Yu and Wang, Ke},
  journal={Proceedings of the ACM on Programming Languages},
  volume={7},
  number={PLDI},
  pages={172--195},
  year={2023},
  publisher={ACM New York, NY, USA}
}

@article{2020-strata,
  title={STRATA: simple, gradient-free attacks for models of code},
  author={Springer, Jacob M and Reinstadler, Bryn Marie and O'Reilly, Una-May},
  journal={arXiv preprint arXiv:2009.13562},
  year={2020}
}

@inproceedings{2020-li-adv-targeted,
	address = {Glasgow, UK},
	author = {Jiawang Bai and Bin Chen and Yiming Li and Dongxian Wu and Weiwei Guo and Shu{-}Tao Xia and En{-}Hui Yang},
	booktitle = {Proceedings of the 16th European Conference on Computer Vision},
	month = {August 23-28},
	pages = {618--634},
	publisher = {Springer},
	title = {Targeted Attack for Deep Hashing Based Retrieval},
	year = {2020}}

@article{2024-CodeLM-Security,
	author = {Yuchen Chen and Weisong Sun and Chunrong Fang and Zhenpeng Chen and Yifei Ge and Tingxu Han and Quanjun Zhang and Yang Liu and Zhenyu Chen and Baowen Xu},
	journal = {CoRR},
	number = {1},
	pages = {1-63},
	title = {Security of Language Models for Code: {A} Systematic Literature Review},
	volume = {abs/2410.15631},
	year = {2024}}

@article{wu2023large,
	author = {Wu, Yonghao and Li, Zheng and Zhang, Jie M and Papadakis, Mike and Harman, Mark and Liu, Yong},
	journal = {arXiv preprint arXiv:2308.15276},
	title = {Large Language Models in Fault Localisation},
	year = {2023}}

@inproceedings{wei2017supervised,
  title={Supervised deep features for software functional clone detection by exploiting lexical and syntactical information in source code.},
  author={Wei, Huihui and Li, Ming},
  booktitle={IJCAI},
  pages={3034--3040},
  year={2017}
}

@article{2023-Literature-Review-Code-Similarity,
	author = {Morteza Zakeri Nasrabadi and Saeed Parsa and Mohammad Ramezani and Chanchal Roy and Masoud Ekhtiarzadeh},
	journal = {Journal of Systems and Software},
	pages = {111796},
	title = {A systematic literature review on source code similarity measurement and clone detection: Techniques, applications, and challenges},
	volume = {204},
	year = {2023}}

@inproceedings{2024-AdaCCD,
	address = {Vancouver, Canada},
	author = {Yangkai Du and Tengfei Ma and Lingfei Wu and Xuhong Zhang and Shouling Ji},
	booktitle = {Proceedings of Thirty-Eighth {AAAI} Conference on Artificial Intelligence},
	month = {February 20-27},
	pages = {17942--17950},
	publisher = {{AAAI} Press},
	title = {AdaCCD: Adaptive Semantic Contrasts Discovery Based Cross Lingual Adaptation for Code Clone Detection},
	year = {2024}}

@inproceedings{2020-Functional-Code-Clone-Detection,
	address = {Virtual Event, USA},
	author = {Chunrong Fang and Zixi Liu and Yangyang Shi and Jeff Huang and Qingkai Shi},
	booktitle = {Proceedings of the 29th {ACM} {SIGSOFT} International Symposium on Software Testing and Analysis},
	month = {July 18-22},
	pages = {516--527},
	publisher = {{ACM}},
	title = {Functional code clone detection with syntax and semantics fusion learning},
	year = {2020}}

@inproceedings{feng2020codebert,
	address = {Online Event},
	author = {Zhangyin Feng and Daya Guo and Duyu Tang and Nan Duan and Xiaocheng Feng and Ming Gong and Linjun Shou and Bing Qin and Ting Liu and Daxin Jiang and Ming Zhou},
	booktitle = {Proceedings of the 25th Conference on Empirical Methods in Natural Language Processing: Findings},
	month = {16-20 November},
	pages = {1536--1547},
	publisher = {Association for Computational Linguistics},
	title = {CodeBERT: {A} Pre-Trained Model for Programming and Natural Languages},
	year = {2020}}

@article{roziere2023code,
  title={Code llama: Open foundation models for code},
  author={Roziere, Baptiste and Gehring, Jonas and Gloeckle, Fabian and Sootla, Sten and Gat, Itai and Tan, Xiaoqing Ellen and Adi, Yossi and Liu, Jingyu and Remez, Tal and Rapin, J{\'e}r{\'e}my and others},
  journal={arXiv preprint arXiv:2308.12950},
  year={2023}
}

@article{2023-Survey-on-LLMs-for-SE,
	author = {Quanjun Zhang and Chunrong Fang and Yang Xie and Yaxin Zhang and Yun Yang and Weisong Sun and Shengcheng Yu and Zhenyu Chen},
	journal = {CoRR},
	number = {1},
	pages = {1--57},
	title = {A Survey on Large Language Models for Software Engineering},
	volume = {abs/2312.15223},
	year = {2023}}

@article{2023-LLMs-for-SE-A-Literature-Review,
	author = {Xinyi Hou and Yanjie Zhao and Yue Liu and Zhou Yang and Kailong Wang and Li Li and Xiapu Luo and David Lo and John C. Grundy and Haoyu Wang},
	journal = {CoRR},
	number = {1},
	pages = {1--62},
	title = {Large Language Models for Software Engineering: {A} Systematic Literature Review},
	volume = {abs/2308.10620},
	year = {2023}}

@article{2022-Use-of-DL-in-SE-Research,
	author = {Cody Watson and Nathan Cooper and David Nader{-}Palacio and Kevin Moran and Denys Poshyvanyk},
	journal = {{ACM} Transactions on Software Engineering and Methodology},
	number = {2},
	pages = {32:1--32:58},
	title = {A Systematic Literature Review on the Use of Deep Learning in Software Engineering Research},
	volume = {31},
	year = {2022}}

@inproceedings{2022-ALERT,
  title={Natural attack for pre-trained models of code},
  author={Yang, Zhou and Shi, Jieke and He, Junda and Lo, David},
  booktitle={Proceedings of the 44th International Conference on Software Engineering},
  pages={1482--1493},
  year={2022}
}

@inproceedings{2020-MHM,
	address = {New York, NY, USA},
	author = {Huangzhao Zhang and Zhuo Li and Ge Li and Lei Ma and Yang Liu and Zhi Jin},
	booktitle = {Proceedings of the 34th {AAAI} Conference on Artificial Intelligence},
	month = {February 7-12},
	pages = {1169--1176},
	publisher = {{AAAI} Press},
	title = {Generating Adversarial Examples for Holding Robustness of Source Code Processing Models},
	year = {2020}}

@inproceedings{guo2020graphcodebert,
	address = {Virtual Event, Austria},
	author = {Daya Guo and Shuo Ren and Shuai Lu and Zhangyin Feng and Duyu Tang and Shujie Liu and Long Zhou and Nan Duan and Alexey Svyatkovskiy and Shengyu Fu and Michele Tufano and Shao Kun Deng and Colin B. Clement and Dawn Drain and Neel Sundaresan and Jian Yin and Daxin Jiang and Ming Zhou},
	booktitle = {Proceedings of the 9th International Conference on Learning Representations},
	month = {May 3-7},
	pages = {1--12},
	publisher = {OpenReview.net},
	title = {GraphCodeBERT: Pre-training Code Representations with Data Flow},
	year = {2021}}

@misc{2024-Our-Artifacts,
  author       = {Tingxu Han},
  title        = {Our Artifacts},
  howpublished = {\url{https://github.com/GeniusHTX/SecCode}},
}

@article{2023-starcoderplus,
  author = {Bigcode},
  title = {StarCoderPlus},
  journal = {Hugging Face Blog},
  year = {2023},
  note = {https://huggingface.co/bigcode/starcoderplus}
}

@article{palacio2024toward,
  title={Toward a theory of causation for interpreting neural code models},
  author={Palacio, David N and Velasco, Alejandro and Cooper, Nathan and Rodriguez, Alvaro and Moran, Kevin and Poshyvanyk, Denys},
  journal={IEEE Transactions on Software Engineering},
  year={2024},
  publisher={IEEE}
}

@inproceedings{gu2018deep,
	address = {Gothenburg, Sweden},
	author = {Xiaodong Gu and Hongyu Zhang and Sunghun Kim},
	booktitle = {Proceedings of the 40th International Conference on Software Engineering},
	month = {May 27 - June 03},
	pages = {933--944},
	publisher = {{ACM}},
	title = {Deep code search},
	year = {2018}}

@inproceedings{sachdev2018retrieval,
  title={Retrieval on source code: a neural code search},
  author={Sachdev, Saksham and Li, Hongyu and Luan, Sifei and Kim, Seohyun and Sen, Koushik and Chandra, Satish},
  booktitle={Proceedings of the 2nd ACM SIGPLAN International Workshop on Machine Learning and Programming Languages},
  pages={31--41},
  year={2018}
}

@inproceedings{leclair2019neural,
  title={A neural model for generating natural language summaries of program subroutines},
  author={LeClair, Alexander and Jiang, Siyuan and McMillan, Collin},
  booktitle={Proceedings of the 2019 IEEE/ACM 41st International Conference on Software Engineering},
  pages={795--806},
  year={2019},
  organization={IEEE}
}

@article{2024-ESALE,
	author = {Chunrong Fang and Weisong Sun and Yuchen Chen and Xiao Chen and Zhao Wei and Quanjun Zhang and Yudu You and Bin Luo and Yang Liu and Zhenyu Chen},
	journal = {IEEE Transactions on Software Engineering (Early Access)},
	pages = {1-18},
	title = {ESALE: Enhancing Code-Summary Alignment Learning for Source Code Summarization},
	year = {2024}}

@inproceedings{haque2020improved,
  title={Improved automatic summarization of subroutines via attention to file context},
  author={Haque, Sakib and LeClair, Alexander and Wu, Lingfei and McMillan, Collin},
  booktitle={Proceedings of the 17th International Conference on Mining Software Repositories},
  pages={300--310},
  year={2020}
}

@inproceedings{wan2019multi,
	address = {San Diego, CA, USA},
	author = {Yao Wan and Jingdong Shu and Yulei Sui and Guandong Xu and Zhou Zhao and Jian Wu and Philip S. Yu},
	booktitle = {Proceedings of the 34th International Conference on Automated Software Engineering},
	month = {November 11-15},
	pages = {13--25},
	publisher = {{IEEE}},
	title = {Multi-modal Attention Network Learning for Semantic Source Code Retrieval},
	year = {2019}}

@article{zeng2023degraphcs,
  title={degraphcs: Embedding variable-based flow graph for neural code search},
  author={Zeng, Chen and Yu, Yue and Li, Shanshan and Xia, Xin and Wang, Zhiming and Geng, Mingyang and Bai, Linxiao and Dong, Wei and Liao, Xiangke},
  journal={ACM Transactions on Software Engineering and Methodology},
  volume={32},
  number={2},
  pages={1--27},
  year={2023},
  publisher={ACM New York, NY}
}

@article{2020-T5,
  author       = {Colin Raffel and Noam Shazeer and Adam Roberts and Katherine Lee and Sharan Narang and Michael Matena and Yanqi Zhou and
                  Wei Li and
                  Peter J. Liu},
  title        = {Exploring the Limits of Transfer Learning with a Unified Text-to-Text Transformer},
  journal      = {Journal of Machine Learning Research},
  volume       = {21},
  pages        = {140:1--140:67},
  year         = {2020},
}

@inproceedings{2019-ASTNN,
  author       = {Jian Zhang and Xu Wang and Hongyu Zhang and Hailong Sun and Kaixuan Wang and Xudong Liu},
  title        = {A novel neural source code representation based on abstract syntax tree},
  booktitle    = {Proceedings of the 41st International Conference on Software Engineering},
  pages        = {783--794},
  publisher    = {{IEEE} / {ACM}},
  address      = {Montreal, QC, Canada},
  month        = {May 25-31},
  year         = {2019},
}

@inproceedings{2016-Convolutional-Neural-Networks-over-Tree-Structures,
  author       = {Lili Mou and Ge Li and Lu Zhang and Tao Wang and Zhi Jin},
  title        = {Convolutional Neural Networks over Tree Structures for Programming Language Processing},
  booktitle    = {Proceedings of the 30th {AAAI} Conference on Artificial Intelligence},
  pages        = {1287--1293},
  publisher    = {{AAAI} Press},
  address      = {Phoenix, Arizona, USA},
  month        = {February 12-17},
  year         = {2016},
}

@inproceedings{mastropaolo2021studying,
  title={Studying the usage of text-to-text transfer transformer to support code-related tasks},
  author={Mastropaolo, Antonio and Scalabrino, Simone and Cooper, Nathan and Palacio, David Nader and Poshyvanyk, Denys and Oliveto, Rocco and Bavota, Gabriele},
  booktitle={Proceedings of the 2021 IEEE/ACM 43rd International Conference on Software Engineering},
  pages={336--347},
  year={2021},
  organization={IEEE}
}

@inproceedings{ahmad2020transformer,
	address = {Online},
	author = {Wasi Uddin Ahmad and Saikat Chakraborty and Baishakhi Ray and Kai{-}Wei Chang},
	booktitle = {Proceedings of the 58th Annual Meeting of the Association for Computational Linguistics},
	month = {July 5-10},
	pages = {4998--5007},
	publisher = {Association for Computational Linguistics},
	title = {A Transformer-based Approach for Source Code Summarization},
	year = {2020}}

@article{2024-LLM4CodeSum,
	author = {Weisong Sun and Yun Miao and Yuekang Li and Hongyu Zhang and Chunrong Fang and Yi Liu and Gelei Deng and Yang Liu and Zhenyu Chen},
	journal = {CoRR},
	number = {1},
	pages = {1--13},
	title = {Source Code Summarization in the Era of Large Language Models},
	volume = {abs/2407.07959},
	year = {2024}}

@inproceedings{2022-Code-Search-based-on-Context-aware-Code-Translation,
  author       = {Weisong Sun and Chunrong Fang and
                  Yuchen Chen and
                  Guanhong Tao and
                  Tingxu Han and
                  Quanjun Zhang},
  title        = {Code Search based on Context-aware Code Translation},
  booktitle    = {Proceedings of the 44th {IEEE/ACM} 44th International Conference on Software Engineering},
  pages        = {388--400},
  publisher    = {{ACM}},
  address      = {May 25-27},
  month        = {Pittsburgh, PA, USA},
  year         = {2022},
}

@inproceedings{kenton2019bert,
  title={Bert: Pre-training of deep bidirectional transformers for language understanding},
  author={Kenton, Jacob Devlin Ming-Wei Chang and Toutanova, Lee Kristina},
  booktitle={Proceedings of naacL-HLT},
  volume={1},
  pages={2},
  year={2019},
  organization={Minneapolis, Minnesota}
}

@article{2023-StarCoder,
	author = {Raymond Li and Loubna Ben Allal and Yangtian Zi and Niklas Muennighoff and Denis Kocetkov and Chenghao Mou and Marc Marone and Christopher Akiki and Jia Li and Jenny Chim and Qian Liu and Evgenii Zheltonozhskii and Terry Yue Zhuo and Thomas Wang and Olivier Dehaene and Mishig Davaadorj and Joel Lamy{-}Poirier and Jo{\~{a}}o Monteiro and Oleh Shliazhko and Nicolas Gontier and Nicholas Meade and Armel Zebaze and Ming{-}Ho Yee and Logesh Kumar Umapathi and Jian Zhu and Benjamin Lipkin and Muhtasham Oblokulov and Zhiruo Wang and Rudra Murthy V and Jason Stillerman and Siva Sankalp Patel and Dmitry Abulkhanov and Marco Zocca and Manan Dey and Zhihan Zhang and Nour Moustafa{-}Fahmy and Urvashi Bhattacharyya and Wenhao Yu and Swayam Singh and Sasha Luccioni and Paulo Villegas and Maxim Kunakov and Fedor Zhdanov and Manuel Romero and Tony Lee and Nadav Timor and Jennifer Ding and Claire Schlesinger and Hailey Schoelkopf and Jan Ebert and Tri Dao and Mayank Mishra and Alex Gu and Jennifer Robinson and Carolyn Jane Anderson and Brendan Dolan{-}Gavitt and Danish Contractor and Siva Reddy and Daniel Fried and Dzmitry Bahdanau and Yacine Jernite and Carlos Mu{\~{n}}oz Ferrandis and Sean Hughes and Thomas Wolf and Arjun Guha and Leandro von Werra and Harm de Vries},
	journal = {CoRR},
	number = {1},
	pages = {1--44},
	title = {StarCoder: may the source be with you!},
	volume = {abs/2305.06161},
	year = {2023}}

@inproceedings{2023-CODA,
  title={Code difference guided adversarial example generation for deep code models},
  author={Tian, Zhao and Chen, Junjie and Jin, Zhi},
  booktitle={2023 38th IEEE/ACM International Conference on Automated Software Engineering (ASE)},
  pages={850--862},
  year={2023},
  organization={IEEE}
}

@inproceedings{wang2016automatically,
  title={Automatically learning semantic features for defect prediction},
  author={Wang, Song and Liu, Taiyue and Tan, Lin},
  booktitle={Proceedings of the 38th international conference on software engineering},
  pages={297--308},
  year={2016}
}

@inproceedings{2020-TextFooler,
  title={Is bert really robust? a strong baseline for natural language attack on text classification and entailment},
  author={Jin, Di and Jin, Zhijing and Zhou, Joey Tianyi and Szolovits, Peter},
  booktitle={Proceedings of the AAAI conference on artificial intelligence},
  volume={34},
  number={05},
  pages={8018--8025},
  year={2020}
}

@article{2020-bertattack,
  title={Bert-attack: Adversarial attack against bert using bert},
  author={Li, Linyang and Ma, Ruotian and Guo, Qipeng and Xue, Xiangyang and Qiu, Xipeng},
  journal={arXiv preprint arXiv:2004.09984},
  year={2020}
}

@inproceedings{gao2018black,
  title={Black-box generation of adversarial text sequences to evade deep learning classifiers},
  author={Gao, Ji and Lanchantin, Jack and Soffa, Mary Lou and Qi, Yanjun},
  booktitle={2018 IEEE Security and Privacy Workshops (SPW)},
  pages={50--56},
  year={2018},
  organization={IEEE}
}

@article{2023-RanMASK,
  title={Certified robustness to text adversarial attacks by randomized [mask]},
  author={Zeng, Jiehang and Xu, Jianhan and Zheng, Xiaoqing and Huang, Xuanjing},
  journal={Computational Linguistics},
  volume={49},
  number={2},
  pages={395--427},
  year={2023},
  publisher={MIT Press One Broadway, 12th Floor, Cambridge, Massachusetts 02142, USA~…}
}

@article{2020-Safer,
  title={SAFER: A structure-free approach for certified robustness to adversarial word substitutions},
  author={Ye, Mao and Gong, Chengyue and Liu, Qiang},
  journal={arXiv preprint arXiv:2005.14424},
  year={2020}
}

@article{zhang2022towards,
  title={Towards robustness of deep program processing models—detection, estimation, and enhancement},
  author={Zhang, Huangzhao and Fu, Zhiyi and Li, Ge and Ma, Lei and Zhao, Zhehao and Yang, Hua’an and Sun, Yizhe and Liu, Yang and Jin, Zhi},
  journal={ACM Transactions on Software Engineering and Methodology (TOSEM)},
  volume={31},
  number={3},
  pages={1--40},
  year={2022},
  publisher={ACM New York, NY}
}

@article{lozhkov2024starcoder,
  title={Starcoder 2 and the stack v2: The next generation},
  author={Lozhkov, Anton and Li, Raymond and Allal, Loubna Ben and Cassano, Federico and Lamy-Poirier, Joel and Tazi, Nouamane and Tang, Ao and Pykhtar, Dmytro and Liu, Jiawei and Wei, Yuxiang and others},
  journal={arXiv preprint arXiv:2402.19173},
  year={2024}
}

@inproceedings{2023-BADCODE,
	address = {Toronto, Canada},
	author = {Weisong Sun and Yuchen Chen and Guanhong Tao and Chunrong Fang and Xiangyu Zhang and Quanjun Zhang and Bin Luo},
	booktitle = {Proceedings of the 61st Annual Meeting of the Association for Computational Linguistics},
	month = {July 9-14},
	pages = {9692--9708},
	publisher = {Association for Computational Linguistics},
	title = {Backdooring Neural Code Search},
	year = {2023}
}

@article{fang2025multi,
  title={Multi-head ensemble of smoothed classifiers for certified robustness},
  author={Fang, Kun and Tao, Qinghua and Wu, Yingwen and Li, Tao and Huang, Xiaolin and Yang, Jie},
  journal={Neural Networks},
  volume={188},
  pages={107426},
  year={2025},
  publisher={Elsevier}
}

@inproceedings{cohen2019certified,
  title={Certified adversarial robustness via randomized smoothing},
  author={Cohen, Jeremy and Rosenfeld, Elan and Kolter, Zico},
  booktitle={international conference on machine learning},
  pages={1310--1320},
  year={2019},
  organization={PMLR}
}

@inproceedings{li2023sok,
  title={Sok: Certified robustness for deep neural networks},
  author={Li, Linyi and Xie, Tao and Li, Bo},
  booktitle={2023 IEEE symposium on security and privacy (SP)},
  pages={1289--1310},
  year={2023},
  organization={IEEE}
}

@inproceedings{zhang2024text,
  title={Text-crs: A generalized certified robustness framework against textual adversarial attacks},
  author={Zhang, Xinyu and Hong, Hanbin and Hong, Yuan and Huang, Peng and Wang, Binghui and Ba, Zhongjie and Ren, Kui},
  booktitle={2024 IEEE Symposium on Security and Privacy (SP)},
  pages={2920--2938},
  year={2024},
  organization={IEEE}
}

@inproceedings{cevher2025certified,
  title={Certified Robustness Under Bounded Levenshtein Distance},
  author={Cevher, Volkan},
  booktitle={The Thirteenth International Conference on Learning Representations},
  year={2025}
}

@inproceedings{ye2020safer,
  title={SAFER: A Structure-free Approach for Certified Robustness to Adversarial Word Substitutions},
  author={Ye, Mao and Gong, Chengyue and Liu, Qiang},
  booktitle={Annual Meeting of the Association for Computational Linguistics (ACL)},
  year={2020}
}
\end{document}